\documentclass[letterpaper,10pt,conference,english]{ieeeconf}
\IEEEoverridecommandlockouts
\usepackage[cmex10]{amsmath}
\usepackage[noadjust]{cite}

\usepackage{%
    amssymb,%
    amsthm,%
    babel,%
    color,%
    enumerate,%
    etoolbox,%
    mathtools,%
    stmaryrd%
}

\usepackage[ruled,vlined]{algorithm2e}
\usepackage{authblk}
\theoremstyle{definition}

\theoremstyle{theorem}

\theoremstyle{prop}
\newtheorem{prop}{Proposition}

\theoremstyle{lemma}

\theoremstyle{remark}

\newcommand{\norm}[1]{\left\lVert#1\right\rVert}

\newcommand{\CE}{\color{black}} % black End of Change
\newcommand{\CR}{\color{black}} %red Need for Addition 
\newcommand{\CG}{\color{black}} %cyan Section Position Changed
\newcommand{\CK}{\color{black}}%blue
\begin{document}

%\title{Performance of Queueing Models for MISO Scheduling in Wireless Content-Centric Networks}
%\title{Performance of Queueing Models for MISO Content-Centric Networks}
%\title{\CG Queueing Theoretic Models for MU-MIMO Content-Centric Networks: Cross-Layer Designs for Joint Queueing and Beamforming}
\title{\CG Queueing Theoretic Models for Multiuser MISO Content-Centric Networks with SDMA, NOMA, OMA and Rate-Splitting Downlink}

\author{Ramkumar Raghu} 
\author{Mahadesh Panju} 
\author{Vinod Sharma}  
\affil{Indian Institute of Science, Bangalore, INDIA. \textit{\{ramkumar,mahadesh,vinod\}@iisc.ac.in}}
\maketitle
\begin{abstract} 
\CG Multiuser, Multiple Input, Single Output (MU-MISO) systems are proving to be indispensable in the next generation wireless networks such as 5G and 6G. The spatial diversity of MISO systems have been leveraged in physical layer designs in these wireless systems to improve the capacity. Several recent studies have utilised redundancies in the content request along with the spatial diversity of a MISO system to improve the capacity further. It is shown that Max-Min Fair (MMF) Beamforming schemes for MISO based on SDMA, NOMA, OMA and Rate-Splitting could be used to improve the content delivery rates. However, in most of these studies the key aspects such as the queueing delays in the downlink and the user dynamics have generally been ignored. In this work, we study how the interplay between queueing, beamforming and the user dynamics affects the Quality-of-Service (user experienced delay) of downlink in MU-MISO content centric networks (CCNs). We propose queueing theoretic models that are simple in nature and can be directly adapted to MU-MISO CCNs to perform optimal multi-group multicast downlink transmissions. We show that a recently developed Simple Multicast Queue (SMQ) for SISO systems can be directly used for MU-MISO systems and that it provides superior performance due to its always-stable nature. Further, we observe that MMF Beamforming schemes coupled with SMQ can be quite unfair to users with good channels. Thus, we propose an improvement to SMQ called Dual SMQ which addresses this issue. We also provide theoretical analysis of the mean delay experienced by the users in such MU-MISO CCNs\CE. %\CE MISO networks have garnered attention in wireless content-centric networks due to the additional degrees of freedoms they provide\CK. Several beamforming techniques based on NOMA, OMA, SDMA and Rate splitting have been proposed for such networks\CE. These techniques utilise the redundancy in the content requests across users and leverage the spatial multicast and multiplexing gains of multi-antenna transmit beamforming to improve the content delivery rate. However, queueing delays and user traffic dynamics which significantly affect the performance of these schemes, have generally been ignored. We study queueing delays in the downlink for several scheduling and beamforming schemes in content-centric networks, with one base-station possessing multiple transmit antennas\CE. These schemes are studied along with a recently proposed Simple Multicast Queue, to improve the delay performance of the network. \CE This work is particularly relevant for content delivery in 5G, 6G and eMBB networks\CE.
\end{abstract}
\begin{keywords}
MISO, Scheduling, Multigroup-Multicasting, Multiplexing, Quality of Service, Queueing Delay, SDMA, NOMA, Rate Splitting, Max-Min Fairness.
\end{keywords}
\section{Introduction}
\CG Proliferation of the Over-the-Top (OTT) platforms such as Netflix, Amazon Prime Video, Youtube etc., has increased the demand for high definition (HD) video contents over wireless networks \cite{CVNI}. It is seen that, depending on the geographic location, a finite subset of the contents hosted by these platforms can sometimes become highly popular \cite{itube}. The wireless base stations (BS) receive repeated requests for these contents from different users in a localised geographic location \cite{Netflix}. Such characteristics of request traffic in the current generation wireless networks have spawned an array of studies focused on Content Centric Networks (CCNs), to leverage such redundancies and improve the network capacity. Further, the broadcast nature of wireless transmitters has been exploited by many of these studies, to provide improvements to both the physical layer and the network layer, by serving the redundant requests in a multigroup-multicast manner. 

In physical layer, multiple antenna, MISO (multiple input, single output) based designs \cite{MMMISO,RS2User,RS2017,MaliMISO,KhalajMISO,precodersurvey} have gained traction due to increased delivery rates they provide in multigroup, multicast transmissions. These schemes utilise high degrees of spatial diversity provided by multiple antennas. The diversity is two fold: First the diversity comes from the broadcast nature of the wireless channel. Multiple antennas in MISO systems accentuate this advantage by enabling simultaneous multicast of different contents. The second diversity comes from the enhanced spatial multiplexing capabilities provided by the multiple antennas. In other words, more than one content can be delivered to the same user (e.g., small base station) simultaneously. Naturally, MISO (and MIMO) has become the cornerstone of 5G, 6G and eMBB (Enhanced Mobile Broadband) networks. 

However many of these works study the physical layer performance in MISO CCNs in isolation. There are other aspects of networks which affect the performance/QoS in a wireless networks such as queueing, user dynamics etc. The importance of cross-layer design with queues in a multiuser MIMO network and the effect of user dynamics was identified in early works, \cite{crosslayer}. However, cross-layer studies of MISO networks, particularly with queues, is seemingly limited, even today. In this paper, we take a holistic approach to the analysis of multigroup-multicasting in a MISO network. We study the interplay between queueing, multiantenna beamforming and user dynamics and the effect it has on the network performance. Unlike, the physical layer studies which study the delivery rate, we concentrate on the user experienced QoS, namely the user's mean delay (sojourn time). In a practical system we will see that the queueing design plays an important role in reducing the delay, as queueing is what controls the scheduling of service transmissions, to a time-varying set of users and their requests.

In the following we provide the relevant literature survey. In \cite{MMMISO}, %which is one of the initial studies in multigroup multicasting, the authors use Max-Min Fair (MMF) beamforming
the authors showthat in a MISO network, that the sum rate maximization, has no QoS or Fairness guarantees. Thus \cite{MMMISO} proposed MMF beamforming that helps in getting QoS (SINR, Rate) guarantees and provide fairness to users with bad channels. Lately, in MISO literature, MMF beamforming has become a norm. Further more, multiple schemes based on non-orthogonal multiple access (NOMA), orthogonal multiple access (OMA), space division multiple access (SDMA) and rate splitting (RS) have also been studied, \cite{RS2User,RS2017, precodersurvey}. Among these schemes, \cite{RS2User} showed that the Rate-Splitting based MMF beamforming has superior performance in terms of content delivery rate (particularly in overloaded conditions). Performance of RS in overloaded Multigroup-Multicast scenario with max-min fairness has been studied in \cite{RS2017} with similar results as in \cite{RS2User}, for a larger system. 
%Rate-Splitting with a common message to all users is studied in \cite{RSCom}. 
%These studies show that RS has similar or better performance (particularly in overloaded condition) than the other beamforming schemes considered in the papers. 
We revisit this conclusion and study MMF beamforming schemes from a cross-layer, queueing perspective. We show that the formulation of RS needs modification in a practical system with queues and that this modification gives a different picture of performance, for different beamforming schemes. 
Rate-Splitting with a common message to all users is studied in \cite{RSCom}. Joint beamforming and coded caching techniques for MISO content centric networks (CCN) has been studied in \cite{MaliMISO,KhalajMISO}. In \cite{MaliMISO} the authors consider caches both at the transmitters and the receivers and propose a joint beamforming that leverages additional degrees of freedom provided by coded caching. Schemes in \cite{KhalajMISO} improve over \cite{MaliMISO} by controlling the group sizes. For a more comprehensive list of beamforming schemes, see \cite{mimosurvey,precodersurvey}.

All of these studies either consider delivery to fixed non-intersecting groups of users or deliver different common message(s) separately to the users from multiple groups. But, in CCNs it is possible that there are time-varying combinations of common users across different groups. We propose adaptations of some of the beamforming schemes mentioned above to cater to this requirement. %Further, a common theme in \cite{RS2User,RS2017,RSCom,MaliMISO,KhalajMISO} and related works is Max-Min Fairness (MMF). 
%Further, as mentioned earlier, none of these works consider the effect of queueing at the BS. 
We show that Max-Min Fairness (MMF), which is a common theme in the above mentioned works, degrades performance for the users with good channels in a queueing system. We address this issue in this paper with an appropriate queueing architecture. 

There are a handful of works which consider queueing in MIMO networks. Queueing Models for Multi-Rate, Multi-User MIMO systems were studied in \cite{mrmu}. The work considers a FIFO queue with finite buffer size and batch service. As a result many user requests go unserviced due to buffer overflow as the traffic increases. In our work, none of the user requests are ever dropped. In \cite{queuestable}, a queueing system with an infinite buffer queue for each user is proposed for MIMO networks. Similar queues for Multiuser MIMO systems with queue-aware scheduling are also proposed in \cite{Qaware}. Queueing specific to SDMA based beamforming schemes is studied in \cite{huang2012stability}. Unlike our system, the proposed queues in these works are not stable across arrival rates due to the nature of queue management. Further, the queues in \cite{huang2012stability,queuestable,Qaware}, are not content-centric, hence, do not leverage the redundancies in the requests. In \cite{delayperform}, the authors analyse delay quantiles in addition to average delay for a multi-user MIMO system. Work in \cite{largedevqueue} performs joint queueing and beamforming decisions in two stages that involve probabilistic selection of users and queue state dependent scheduling. It is shown that this two stage operation improves queue length and channel feedback performance. However, \cite{delayperform} and \cite{largedevqueue}, do not consider multi-group, multicast transmissions as in our work and hence, are inherently not efficient for CCNs. Also, the queues in these works are also not always stable. In fact, instability of queues in these queueing systems arises due to a common drawback. All these works have different queues for different users which reduces the multicast opportunities enabled by the redundancy in the content. Infact we have shown that such queues have inferior performance in the content-centric networks, as seen in our previous work \cite{TWC2021}. Hence, it behoves to design queueing systems, that are content-centric in nature, even for a MIMO system. That is, the queue(s) must exploit the redundancies in the content and perform content dependent scheduling.

To the best of our knowledge, the only alternative to \cite{TWC2021}, for content centric queues is \cite{zhou2017optimal}. In \cite{zhou2017optimal} a dynamic multicast scheduling is proposed, which minimizes queueing delay and power using reinforcement learning (RL). Scheduling in \cite{zhou2017optimal} is state-dependent and is not scalable with system size due to computational complexity. Hence it is not suitable for MIMO systems. We point out that, though the queueing system in \cite{zhou2017optimal} is content centric, the stability across all arrival rates is not guaranteed. Further, \cite{TWC2021, Arxiv2019} propose simple multicast queueing schemes and reinforcement learning based power control, which are shown to outperform other similar queueing schemes in literature. In addition, the Simple Multicast Queue (SMQ) proposed in \cite{TWC2021} is analytically guaranteed to be always stable across any arrival rate. 

Readers are referred to \cite{mimosurvey} for a comprehensive list of queueing schemes for MIMO systems.%Thus the SMQ serves as a good candidate for MISO Multigroup Multicast System considered in this paper. We also note that there are other works, \cite{mimosurvey}, which propose alternate fairness criteria for physical layer beamforming\CE. However, we will see that MMF beamforming is the optimal scheme for SMQ\CE.

Since, the Simple Multicast Queue (SMQ) proposed for SISO networks is simple, scalable, content-centric and stable across all arrival rates, it is natural to consider it for MISO setup as well. However, in MISO systems with MMF Beamforming, when user channel statistics are heterogenous (where %some users have good channel statistics and the others have bad channel statistics, 
there are different sets of users with good and bad channel statistics), we will see that the adaptations of SMQ in \cite{TWC2021} such as Loopback, Defer \cite{panju19}, and Power control in time \cite{Arxiv2019} are no longer directly useful. % for MMF beamforming schemes. 
We will show that though this combination of MMF schemes and SMQ, may provide good overall mean delay for the network, it may penalise good channel users in a heterogenous network. Thus, in this paper, we propose modifications to SMQ to address this issue.  

Following are our main contributions in this work:  
%In this paper we further explore the performance of SMQ in MISO setup. Particularly we make following important  contributions in this work:
\begin{itemize}
\item We consider a recently proposed SMQ \cite{TWC2021} and show that it can be directly adapted to MISO CCNs. The aim here, is to provide a queueing architecture for MISO systems which is easy to implement and computationally simple, yet provide significantly better performance than existing solutions.
\item We study various beamforming schemes studied in literature \cite{RS2017,RS2User} from a cross-layer, queueing perspective. We provide necessary adaptations to these beamforming schemes, required for queueing and transmission to time-varying and common users across different multicast groups. We show that the performance of RS differs significantly from \cite{RS2017,RS2User} in such a setup.%We provide unique observations on the performance of these schemes in wireless content-centric networks with queues.
\item We show that, in a heterogenous network, MMF beamforming schemes can be pretty unfair to the users with good channel statistics. We propose a novel two queue architecture, named Dual Simple Multicast Queue (DSMQ), to address this issue. We show that DSMQ, provides a decoupled, fair QoS to users with good channels, in a network with heterogeneous channels. The users with bad channel gains also may not suffer much.
\item We provide theoretical analysis and insights to the proposed queues. 
\begin{itemize}
  \setlength{\itemindent}{0em}
  \item First, we prove the stationarity of SMQ and DSMQ and show that they are always stable even in MISO setup.
  \item Second, we derive queueing theoretical approximations to the user experienced delay (mean sojourn time) in a Multiuser MISO CCNs with queues and MMF beamforming. These approximations serve as tools to system engineers, to evaluate the cross-layer aspects of queueing and beamforming, without computationally expensive and time consuming simulations.  
\end{itemize}

\item Finally, we present extensive performance study via simulations, for the queueing and MMF beamforming schemes considered in this paper. Here, we also show that for Multiuser MISO CCNs, the schemes developed for SISO systems such as Loopback, Defer \cite{panju19,TWC2021} and reinforcement learning based Power control in time \cite{Arxiv2019} are not useful. This is because, the time diversity utilised by these schemes in a SISO system is compensated by the spatial diversity of the MISO system. %This is a great advantage of multiple antennas, particularly with power control in time, because the RL power control \cite{Arxiv2019}, is computationally very expensive as it needs channel gains at the BS for multiple antennas.
%\item We show via simulation that \CE reinforcement learning \CE based power control in time, which was very useful in single antenna case \cite{Arxiv2019}, provides no gain for multi-antenna systems if there are sufficient number of antennas. This is a great advantage of multiple antennas because power control, though useful, needs channel gains at the BS and can be computationally very expensive, especially for multiple antennas\CE. 

%\item We propose an extension of SMQ to provide differentiated QoS in networks where users with both good and bad channel are present. 
%\item We show that appropriate power allocation in time under average power constraint, helps improve the performance of the MISO system.
\end{itemize}

Rest of the paper is organised as follows. Section \ref{sec:system_model} describes the system model, assumptions and SMQ. Section \ref{sec:beamf} presents the beamforming schemes considered in this paper. Section \ref{sec:DSMQ} describes a new type of queue called Dual Simple Multicast Queue (DSMQ) which improves over SMQ in terms of fairness. Section \ref{sec:station} provides proof of stationarity for SMQ and DSMQ. In Section \ref{sec:smq_theory} we present the theoretical approximations for mean delay (sojourn time) for SMQ and DSMQ in the MISO system. Section \ref{sec:sim} provides simulation results and compares different schemes. Finally, Section \ref{sec:conclusion} concludes the paper.

\textit{Notation: $\{.\}^H$, $\{.\}^T$ represent Hermitian and Transpose operations respectively, $[N]$ represents the set of natural numbers upto $N$, $\norm{.}_2$ represents $\mathcal{L}_2$ norm, $diag(\textbf{g})$ represent a diagonal matrix formed by elements of vector $\textbf{g}$.}\CE
\section{System Model}
\label{sec:system_model}
We consider a wireless content-centric network with one BS endowed with $L$ transmit antennas and $K$ user equipments (UE). The system model is shown in Fig. \ref{fig:MISO}. Each UE requests contents from a library of $N$ files. Each file is of size $F$ bits and the receiver bandwidth is, $B$ MHz. In practical networks the UEs can be either a mobile user or a small BS (SBS). 
The channel between transmit antennas and a UE follows flat fading. In other words the channel stays constant for the duration of each file transmission and independently changes in the next transmission. 
The request traffic from each user follows Independent Reference Model (IRM). In IRM, the request process of each content $n\in[N]$ from each UE $k$, is an independent Poisson process with rate $\lambda_{nk}$. The overall rate of request traffic to the BS is given as $\lambda=\sum_{n,k}\lambda_{nk}$. These requests are queued and served by the BS using SMQ, \cite{TWC2021}.%Given a user $k\in[K]$, the $\lambda_{nk}$ is defined by a Zipf distribution with parameter $\gamma_k$ on content library. That is if $\lambda_k$ is the traffic from user $k$, then, $\lambda_{nk}=p_{nk}\lambda_k$, where $p_{nk}\propto 1/n^{\gamma_k}$. However, we note that the popularity is not restricted to Zipf and can be more general. For simplicity we assume $\gamma_k=\gamma$ and $\lambda_k=\lambda/K$. 
\begin{figure}[!h]
\centering
\includegraphics[height=6.5cm,width = 7cm]{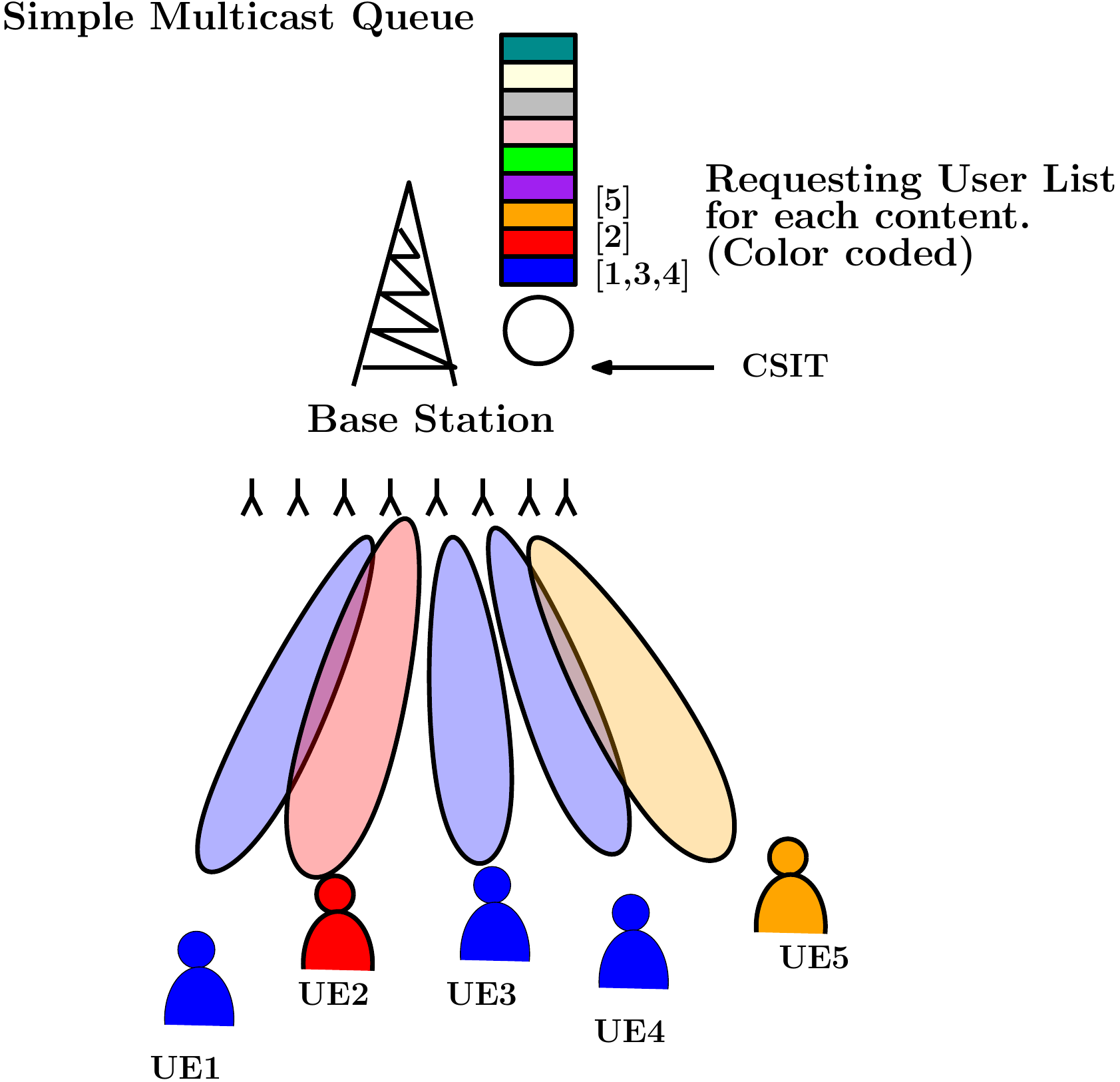}
\caption{Simple Multicast Queue system with Multi Antenna Base Station.} %The figure shows a typical scenario where multiple groups of users request different files (files are show in different colors). The UEs requesting a particular content are listed in the corresponding location of the content in the queue. The BS may choose to transmit one or more files (three in this example) to the requesting groups.}
\label{fig:MISO}
\end{figure}

\CK\textbf{Simple Multicast Queue (SMQ):}   In this queue, the files are served in a First-Come-First-Serve fashion with a slight modification: \textit{\CK when a file is transmitted, it serves all the users who have requests for it in the queue}. At a given instance the $j^{th}$ entry of the queue is denoted by $(n,\mathbb{L}_n)_j$, where $n\in[N]$ is the file index, $\mathbb{L}_n\neq\{\phi\}$ is the list of users requesting file $n$. When a new request for a file $n$ from user $k$ arrives at SMQ, it is merged at $j^{th}$ location as $\mathbb{L}_n=\mathbb{L}_n\bigcup k$. If none of the entries of SMQ has the file $n$ then a new entry $(n,L_n=\{k\})$ is added to the tail of the queue. %A new request for a file, from a UE is merged with the previous requests of the same file in the queue. This way multiple requests from multiple users for the same file are merged as one entry in the queue. %The queue serves the merged requests in the head of the line at each transmission, simultaneously. 
 %Using regenerative arguments it is also shown that the system has a unique stationary distribution. Using M/G/1 approximations, an approximate mean delay formula is also provided in \cite{TWC2021}. In this paper we further analyse the performance of the SMQ in multi-antenna case. 
The BS simultaneously transmits first $\min(s,S)$ files starting from the head of the line of SMQ using one of the schemes in Section \ref{sec:beamf}, where $s$ is the queue length and $S$ is a configurable system parameter. New requests for a file under service are added to the tail of SMQ and the subsequent requests are merged as before. Thus the queue length does not exceed $N$ at any given time. Further, the \textbf{\textit{sojourn time}}, $D$ of a request arriving at the BS at time $t_r$ and serviced at time, $t_c$ is given by, $D=t_c-t_r$. Here $D$ includes the service time $T^*$ (defined for each scheme in Section \ref{sec:beamf}). %In Section \ref{sec:station} we prove the existence of a unique stationary distribution of $D$ for queues considered in this paper. 
The mean sojourn time under stationarity (see Section \ref{sec:station}) is denoted as $E[D]$\CE. 

In the following sections we describe the beamforming strategies considered in this paper\CE. We assume complete channel state information at the transmitter (CSIT). In each channel use, the channel matrix $\textbf{H}\in \mathbb{C}^{L\times K}$ is drawn independently as $\textbf{H} \sim \mathcal{CN}(0,\Sigma )$, the complex Gaussian distribution with mean $0$ and covariance $\Sigma$. We consider two cases, namely, \textit{homogeneous} channels and \textit{heterogeneous} channels. In homogenous channel case, $\Sigma=g \textbf{I}$, where $\textbf{I}$ is the identity matrix of size $KL\times KL$, and $g$ is the mean of channel gain of each channel when all the users have similar channel statistics. In a heterogenous network different users channel statistics may be different. Here, $\Sigma=diag(\textbf{g}_{1}^T,\cdots,\textbf{g}_{K}^T)$, where, $\textbf{g}_{k}=g_k \textbf{1}$, $g_k$ is the mean fading (or) channel gain of user $k$ and $\textbf{1}$ is the vector of all ones of size $L$\CE. Further the transmitter may choose to transmit one or more files from the queue depending on the beamforming strategy. $S$ files in the head of the line of SMQ are denoted by $\mathcal{X}_S\triangleq\{X_1,X_2,...,X_S\}$. We assume that the transmitter uses Gaussian codebook to assign a codeword $\tilde{X}_s$ for each file $X_s$ and the codebook is known to the receiver as well \cite{KhalajMISO}. %Each file $X_s$ is mapped to codeword $\tilde{X}_s$, $s\in [S]$ using a unit power complex Gaussian codebook \cite{KhalajMISO}. 
Since SMQ merges different requests across users for the same content, it is possible that a user might have requested more than one content in $\mathcal{X}_S$.
\section{Beamforming Schemes}
\label{sec:beamf}
We now describe the \CK SDMA, OMA, NOMA and RS based beamforming schemes used in this paper. For $S=1$, schemes in Sections \ref{sec:MMF} and \ref{sec:MMF_SIC} are same and correspond to OMA. For $S\geq 2$, the scheme in Section \ref{sec:MMF}, corresponds to SDMA when there are no common users and otherwise, a combination of NOMA (for common) and SDMA (for other users). For $S\geq 2$, Section \ref{sec:MMF_SIC} gives a NOMA based scheme. Finally, Section \ref{sec:MMF_RS} presents our RS based scheme. These beamforming schemes are designed for queueing, time varying sets of active users and also cater to common users across multiple groups\CE.
%\subsection{Max-Min Fair (MMF) Multigroup Multicast Beamforming:}
\subsection{Max-Min Fair (MMF) Beamforming:}

\label{sec:MMF}
\CK Let the subset $\mathcal{X}\subset\{X_1,\cdots,X_S\}$, be the set of files requested by user $k$. In this strategy the transmitter chooses a precoding vector $\textbf{w}_s\in \mathbb{C}^{L\times1}$ for each file $s\in[S]$ from $S$ head of the line files. Thus received signal $y_k$ at user $k$ is given as
\begin{equation}
y_k=\sum_{s:X_s\in\mathcal{X}} {\textbf{h}_k^H}\textbf{w}_s\tilde{X_s}+\sum_{t:X_t\in\mathcal{X}_S\setminus\mathcal{X}} {\textbf{h}_k^H}\textbf{w}_t\tilde{X_t}+n_k,
\end{equation}
where  $k\in\mathcal{U},\ n_k\sim \mathcal{CN}(0,\sigma_k^2)$, $\textbf{w}_s\in \mathcal{C}^{L\times1}$, is the spatial precoding vector selected by the transmitter for file/stream $s$ and $\textbf{h}_k\in \mathcal{C}^{L\times1}$ is the $k^{th}$ column of matrix $\textbf{H}$.  The SINR of transmitted file $s \in \mathcal{X}$, requested by the $k^{th}$ user is: 

\begin{equation}
\gamma_k^s=\frac{{\vert\textbf{h}_k^H}\textbf{w}_s\vert^2}{\sigma_k^2+\sum_{t\in\mathcal{X}_S\setminus\mathcal{X}}{\vert\textbf{h}_k^H}\textbf{w}_t\vert^2}.
\end{equation}

Let $\mathcal{U}_s$ be the set of users requesting file $s$ and let $\mathcal{U}_A$ be the set of users requesting subset $A\subset[S]$. Further let $\mathcal{A}$ be the collection of all the subsets of $A\subset[S]$ with cardinality greater than one. That is, if $\tilde{B}\in\mathcal{A}$, then $\vert \tilde{B} \vert >1$.  %let $\mathcal{U}^n_I$ be the subset of users requesting subset $I\in\mathcal{I}^n$ of $\mathcal{X}_S$, consisting of $n$ files. Also, $\mathcal{I}^n$ is the collection of all subsets of $\mathcal{X}_S$ of size $n$. 

The precoding weights $\textbf{w}_s, s\in[S]$ are obtained by solving the following optimization problem $P1$:
\begin{equation}
\label{eq:opt}
\centering
\begin{split}
\max_{R_k^s, \textbf{w}_s, s\in[S]}\ \min_{k\in\mathcal{U}_s, s\in[S]} R^s_k\ \ \ \ \ \ \ \ \ \ \ & \ \ \\
%\text{such that}\ \ \ \ \ \ \ \  \ \ \ \ \ \ \ \ \  \ \ \ \  \ \ \ \ \ \ \ \ \ \ \ \ \ \ \ \ \ \ \ \ \ \ \ \ \ \ \ \ \ &\ \ \ \\
\text{s.t. } R_k^s\leq\log_2\left( 1+ \gamma_k^s\right), \forall k\in\mathcal{U}_s,\ s\in[S],\ \ \\
%\forall k\in\mathcal{U}_s,\ s\in[S],\ \ \ \ \ \ \ &\\
%\sum_{j\in\mathcal{M}_m}R_u^j\leq&\log_2(1+\sum_{j\in\mathcal{M}_m}\gamma_u^j),\\
\sum_{j\in \tilde{B}}R_u^j\leq\log_2(1+\sum_{j\in \tilde{B}}\gamma_u^j),\ \ \ \ \ \ \ \\
\text{for all $\tilde{B}\in\mathcal{A}$, $u\in\mathcal{U}_A$, } A\subset [S],\text{ and}\ \ \ \ \ \ \ \ \ \ \ \ \ \ \ \ \   \\
%\forall u\in\mathcal{U}_I^n,\ I\in\mathcal{I}^n,\ n\in&[S]\setminus\{1\},\\
%\forall\mathcal{M}_m:\mathcal{M}_m\subset I,\ \vert&\mathcal{M}_m\vert=m, m=[n]\setminus\{1\},\\
\sum_{s\in[S]}\norm{\textbf{w}_s}_2^2 \leq P.\ \ \ \ \ \ \ \ \ \ \ \  \  \ \ 
\end{split}
\end{equation}
The first set of inequalities are the rate constraints based on the Gaussian capacity for a single stream, considering the unwanted streams as noise. The users who want more than one file (common users), decode the required files using Successive Interference Cancellation (SIC). The second set of inequalities are for SIC MAC constraints for the common users or users who have more than one file requests in the queue. Finally, the last inequality is the total power constraint with power $P$. %We denote the ratio of transmit power to noise power as $SNR = P/N_0$.

Since $S$ files are being transmitted simultaneously, the total transmit time, $T^*=F/(R^*B)$, where $R^*$ is the optimal rate (in bits/sec/Hz) obtained by solving $P1$\CE.   
%This MMF problem is known to be NP-Hard \cite{KhalajMISO}. However good sub-optimal %points can be obtained using methods like Successive Convex Approximation (SCA), and reformulations as Second Order Cone Problem (SOCP) \
%algorithms such as SCA, SOCP are available \cite{KhalajMISO}. %Since the objective of this paper is to bring out the optimal beamforming strategy for queueing, we do not get into details of these reformulations. 
%We use Python's, SciPy SLSQP, which gives better performance than SOCP in our simulations.%, directly to solve our optimization problems. Infact, we have seen in our simulations that SLSQP gives better performance than SOCP \cite{KhalajMISO}, for $S=1$.%Infact, we have seen that SLSQP offers better optimal points compared to SOCP, \cite{KhalajMISO}, implemented using CVXPY \cite{cvxpy}, for $S=1$. \CE 
%%See extended version of this paper \cite{report} for symmetric rate reformulations of problems $P1-P3$\CE. 

%%Modified by Panju
\subsection{Max-Min Fair Beamforming with Full SIC (MMF-SIC)} 
\label{sec:MMF_SIC}
\CE In this strategy $S$ head of the line requests, $\mathcal{X}_S$  are transmitted to all the users  $\overline{\mathcal{U}}=\underset{s\in[S]}{\bigcup}\mathcal{U}_s$. Thus SINR for stream, $s$ at user, $k$ is $\overline{\gamma}^s_k={{\vert\textbf{h}_k^H}\textbf{w}_s\vert^2}/{\sigma_k^2}$.
All the messages are decoded at all the users using SIC. Let $R^s_k$ be the rate alloted to stream $s$ for user $k$, $\forall$ $s\in[S]$ and $k\in\overline{\mathcal{U}}$. Since all the users receive all the files, we have the flexibility to consider $S$ files in $\mathcal{X}_S$ as a single file of size $SF$ and rearrange the sizes of files to be transmitted in different streams as $\mathcal{X}^{\beta}_S=\{X^{\beta_1}_1,\cdots,X^{\beta_S}_S\}$, where the file $X^{\beta_s}_s$ is of size $\beta_s S F$, $\beta_s\in[0,1]$, $s\in[S]$ and $\sum_{s\in[S]}\beta_s=1$. Thus transmit time of stream $s$ to user $k$ is given by $T^s_k=\frac{\beta_s SF}{B R^{s}_k}$. Therefore, to minimize the transmit time of $S$ files, we have the following optimization problem $P_2$:
\begin{equation}
\label{eq:opt_SIC}
\centering
\begin{split}
%\min_{R^{s}_k,\beta_s, \textbf{w}_s, \forall s\in[S]\}} \max_{s \in [S]} \frac{\beta_s SF}{R^{(s)}}\\
\min_{R^{s}_k,\beta_s, \textbf{w}_s, s\in[S]} \ \max_{k\in\overline{\mathcal{U}}, s \in [S]}{T^s_k}\\
\text{s.t. } 
%\sum_{s \in [S]} \beta_s = 1,\\
\sum_{s \in \overline{\mathcal{S}}} R^s_k  \leq \log_2( 1+ \sum_{s \in \overline{\mathcal{S}}} {\overline{\gamma}^s_k}),& \forall  \overline{\mathcal{S}}\subset [S],  k \in \overline{\mathcal{U}},\\
%\forall  \overline{\mathcal{S}}\subset [S],  k \in \overline{\mathcal{U}}, \ &\\
\sum_{s \in [S]} \beta_s = 1,\ \sum_{s\in[S]}\norm{\textbf{w}_s}_2^2 \leq& P.
\end{split}
\end{equation}
%In the above setup, there are $S$ streams, with rate $R^s$ for $s \in [S]$. 
The channel from the BS to a given user $k$ forms a MAC channel with $S$ messages. The first set of constraints ensure $R^s_k$ for $s \in [S]$ and $k \in \overline{\mathcal{U}}$, lie in the achievable region for every user. This ensures that every user can decode all the $S$ streams using SIC\CE. The second equality constraint on file size fraction ensures that all the $SF$ bits are split to $S$ streams. The last inequality is the total power constraint. %In this strategy, the $S$ head of the messages, $\mathcal{X}_S$ are treated as one unit of total size $SF$ bits and transmitted to all the users $\underset{s\in[S]}{\bigcup}\mathcal{U}_s$. We form $S$ streams of size $\beta_s SF$ bits for $s \in [S]$ and $\sum_{s \in [S]} \beta_s = 1$. Each user decodes all the $S$ streams using SIC. Let $\mathbf{S} \subset [S]$. The time taken for each stream is $\beta_s SF/R^S$ where $R^s$ is the rate assigned for stream $s$ for $s \in [S]$. \\
%We have the following optimization problem $P_2$:
%\begin{equation}
%\label{eq:opt_SIC}
%\centering
%\begin{split}
%\max_{R^s, \textbf{w}_s, \forall s\in[S]} \sum_{s \in [S]}  & R^s \\
%\text{s.t. } 
%\forall  \mathbf{S}\subset [S]  \text{ and } k \in \underset{s\in[S]}{\bigcup}\mathcal{U}_s \ &\\
%\sum_{s \in \mathbf{S}} R^s & \leq \log_2\left( 1+ \sum_{s \in \mathbf{S}} {\textbf{h}_k^H}\textbf{w}_s\right) \\
%\end{split}
%\end{equation}
%In the above, there are $S$ streams with rate $R^s$ for $s \in [S]$. 
%The channel from BS to a given user $k$ forms a MAC channel with $S$ messages. 
%The constraints ensure $R^s$ for $s \in [S]$ lie in acheivable region of every user. 
%In other words, the rate of $S$ streans lie in intersection of acheivable region of all users. 
%Thus, every user can decode all the $S$ streams. 
%The individual rate within this intersection of acheivable regions is chosen in such a way all $S$ require the same amount of time $T$ to transmit. 
%Furhtermore, we minimize the time for transmission $T$. 
%This is done as follows. 
%We treat $S$ files as one unit of size $SF$ bits. 
%If we split $SF$ bits into $S$ parts such that the parts have size $R^s SF/(\sum_{s'} R^{s'})$, the transmission time $T = SF/(\sum_{s'} R^{s'})$ is the same for each part. 
%To minimize $T$, we maximize $\sum_{s'} R^{s'}$ as indicated in the objective function\CE.
%\vspace{-10pt}
\subsection{Max-Min Fair Rate Splitting (MMF-RS) Beamforming}
\label{sec:MMF_RS}
In this section we give a new formulation of Rate Splitting (RS) proposed in \cite{RS2017}. %Before we present our formulation we briefly discuss the Rate Splitting proposed in \cite{RS2017} and bring out the need for new formulation. Classically RS was proposed for SISO system \cite{RSOrigin}. 
In RS \cite{RS2017}, the idea is to split a particular file into two parts, map the parts to two symbols and transmit both the parts simultaneously with different precoding weights. At the receiver the first symbol is decoded considering the other part as interference and then the decoded symbol is cancelled from the received signal and the second signal is recovered. The optimal weights are obtained by optimizing the sum rates of both the streams. While this is a good objective for physical layer, the queueing layer has to wait for transmission of both the files before the next can be served\CK. Thus instead of maximizing the sum rate, we need to minimize the maximum transmit time for both the parts. For our MISO case, the problem is formulated as follows\CE:

As in section \ref{sec:MMF}, we consider $S$ files at the head of the line of the queue. Each file $X_s\in\mathcal{X}_S$ of size $F$ bits is split into two parts $X^0_s$ and $X^1_s$ with corresponding sizes $\alpha_s F$ and $(1-\alpha_s) F$ correspondingly, where $\alpha_s\in[0,1]$ is an optimization parameter for file $s$. The transmitter chooses weights $\textbf{w}_s, s\in[S]$ for transmitting mapped symbols $X^1_s\rightarrow\tilde{X}^1_s, s\in[S]$. The parts $\{X^0_1,\cdots,X^0_S\}$ are mapped to a single symbol as $\{X^0_1,\cdots,X^0_S\}\rightarrow\tilde{X}^0_D$. The subscript $D$ represents degraded transmission as in \cite{RS2017}. %It is called degraded transmission because all the file subparts are decoded by all the users, hence this incurs a Degrees of Freedom (DoF) loss \cite{RS2017}. 
The SINR for the degraded stream at user $k$ is given as:
\begin{equation}
\gamma^D_k=\frac{{\vert\textbf{h}_k^H}\textbf{w}_D\vert^2}{\sigma_k^2+\sum_{s\in[S]}{\vert\textbf{h}_k^H}\textbf{w}_s\vert^2}.
\end{equation}
Let $R_s$ be the rate allocated to stream $s$ and $R_D$ be the rate allocated to stream $D$. The transmitter chooses precoder $\textbf{w}_D$ for transmitting $\tilde{X}^0_D$. The transmission times of symbols  $\tilde{X}^1_s, s\in[S]$ and $\tilde{X}^0_D$ are $(1-\alpha_s)F/(R^s B), s\in[S]$ and $\sum_{s\in[S]}\alpha_sF/(R^D B)$ respectively. Now we want to minimize the maximum transmit time. This leads to the following optimization problem $P3$:  
\begin{equation}
\label{eq:opt_RS}
\centering
\begin{split}
\min_{\substack{R^D,\textbf{w}_D,R^s_k, \textbf{w}_s\\ 0\leq\alpha_s\leq1, s\in[S]}}\ \max_{k\in\mathcal{U}_s, s\in[S]} \{\frac{\sum_{s\in[S]}\alpha_sF}{BR^D },\frac{(1-\alpha_s)F}{BR^s_k}\} \\
%\text{where the min is over } {R^D,\textbf{w}_D,R^s_k, \textbf{w}_s  0\leq\alpha_s\leq1, s\in[S]}\\
%\text{such that }\ \ \ \ \ \ \ \ \ \ \ \ \ \ \ \ \ \ \ \ \ \ \ \ \ \ \ \ \ \ \ \ \ \ \ \ \ \ \ \ \ \ \ \ \ \ \ \ \ \ \ \ \\
\text{s.t.\ } R^D\leq\log_2\left( 1+ \gamma_k^D\right), \forall k\in\mathcal{U}_s,\ s\in[S],\ \ \ \ \ \\
%\forall k\in\mathcal{U}_s,\ s\in[S],\ \ \ \ \ \ \ \ \ \ \ \ \ \ \ \ \ \ \ &\\
%R_k^s\leq&\log_2\left( 1+ \gamma_k^s\right)\\
%0\leq\alpha_s\leq1, \forall s\in[S]&\\
%\forall u\in\mathcal{U}_I^n,\ I\in\mathcal{I}^n,\ n\in&[S]\setminus\{1\}\\
%\sum_{j\in\mathcal{M}_m}R_u^j\leq&\frac{1}{m}\log_2(1+\sum_{j\in\mathcal{M}_m}\gamma_u^j)\\
%\forall\mathcal{M}_m:\mathcal{M}_m\subset I,\ \vert&\mathcal{M}_m\vert=m, m=[n]\setminus\{1\}\\
\norm{\textbf{w}_D}_2^2+\sum_{s\in[S]}\norm{\textbf{w}_s}_2^2 \leq P,\ \ \ \ \ \ \ \ \ \ \ \ \ \ \ \\
\text{and rest of the constraints, as in (\ref{eq:opt})}. \ \ \ \ \ \ \  \ \ \ \ \ \ \ \ \ \ \ \ \ \
\end{split}
\end{equation}
The constraint on $R_D$ ensures delivery of the degraded stream to all the users. Finally the last inequality is the total power constraint\CK. In both, $P2$ and $P3$, the service time $T^*$ is the optimal value of the respective objective function\CE.  % at $t^{th}$ service instant. The service time for transmitting all the $S$ files, $s_t=T^{*}_t$. 

\subsection{Optimization Reformulation and Service Time:}
\label{sec:reformulate}
\CK For a tractable queueing system, the transmitter serving multiple groups simultaneously should start the next transmission after all the transmissions are complete. Thus imposing symmetric rate across all groups (i.e., transmitting all groups with the optimal min rate), does not change the optimal transmit time obtained by solving $P1-P3$. Thus by imposing symmetric rate requirement our problem $P1$ can be reformulated as:
\begin{equation}
\label{eq:opt_re}
\centering
\begin{split}
\max_{r,\textbf{w}_s, s\in[S]}r\ \ \ \ \ \ \ \ \ \ \ \ \ \ \ \ \ \ \ \  \ \ \ \ \  \ \  \ \ \\%\min_{k\in\mathcal{U}_s, s\in[S]} &r \\
\text{s.t. }
r\leq\log_2\left( 1+ \gamma_k^s\right),\ \forall k\in\mathcal{U}_s,\ s\in[S],\ \ \ \ \ \ \ \ \ \ \ \ \ \ \ \ \ \ \ \ \ \ \ \  \\
r\leq\frac{1}{\vert\tilde{B}\vert}\log_2(1+\sum_{j\in\tilde{B}}\gamma_u^j),\ \ \ \ \ \ \ \ \ \ \ \ \ \ \ \  \\
\text{for all $\tilde{B}\in\mathcal{A}$, $u\in\mathcal{U}_A$, } A\subset [S],\text{ and}\ \ \ \ \ \ \ \ \ \ \ \ \ \ \ \  \ \ \ \ \ \ \ \ \ \ \ \  \\
%\forall u\in\mathcal{U}_I^n,\ I\in\mathcal{I}^n,\ n\in&[S]\setminus\{1\},\\
%\forall\mathcal{M}_m:\mathcal{M}_m\subset I,\ \vert&\mathcal{M}_m\vert=m, m=[n]\setminus\{1\},\\
\sum_{s\in[S]}\norm{\textbf{w}_s}_2^2 \leq P.\ \ \ \ \ \ \ \ \ \ \ \ \ \ \ \ \ \ \ \ \ \ \ \  \ \ \ 
\end{split}
\end{equation}
The optimization $P2$ and $P3$ however require a slightly different reformulation when we impose a symmetric rate constraint. 
In $P_2$, we impose that the time taken to transmit each stream to be the same. Towards this we consider $R^{s}=\min_{k\in\overline{\mathcal{U}}} R^s_k$ for all $s \in [S]$. Further, for achieving same transmit time, we set the fraction $\beta_s = R^{s}/(\sum_{t\in[S]} R^{t})$, $\forall s \in [S]$. The symmetric transmission time is thus, $T = SF/(B\sum_{s} R^{s})$ for each stream. We define the symmetric rate of transmission as $r=1/T$. Thus $P2$ can be reformulated as maximization of $\sum_{s} R^{s}$. This leads us to the following simplified formulation of $P2$:

\begin{equation}
\label{eq:opt_SIC_re}
\centering
\begin{split}
\max_{R^{s}, \textbf{w}_s, \forall s\in[S]} \sum_{s \in [S]} R^s \ \ \ \ \ \ \ \ \ \ \ \ \ \ \ \ \\
\text{s.t. } 
\forall  \overline{\mathcal{S}}\subset [S]  \text{ and } k \in \overline{\mathcal{U}},\ \ \ \ \ \ \ \ \ \ \ \ \ \ \ \ \\
\sum_{s \in \overline{\mathcal{S}}} R^s  \leq \log_2\left( 1+ \sum_{s \in \overline{\mathcal{S}}} {\textbf{h}_k^H}\textbf{w}_s\right), \ \ \ \ \ \ \ \ \ \ \\
\sum_{s\in[S]}\norm{\textbf{w}_s}_2^2 \leq P.\ \ \ \ \ \ \ \ \ \ \ \ \ \ \ \ \ \ \ \ 
\end{split}
\end{equation}

Similarly, in $P3$ we impose $r=\frac{BR^D}{\sum_{s\in[S]}\alpha_sF }=\frac{BR^s_k}{(1-\alpha_s)F}$ for all feasible $s,k$'s. This leads to the following reformulation of the optimization problem $P3$:
\begin{equation}
\label{eq:opt_RS_re}
\centering
\begin{split}
\max_{r, \textbf{w}_D,\textbf{w}_s, s\in[S]}r\ \ \ \ \ \ \ \ \ \ \ \ \ \ \ \ \ \ \ \ \ \ \ \ \ \ \ \ \ \ \\%\min_{k\in\mathcal{U}_s, s\in[S]} &r \\
\text{s.t.}\ \ \forall k\in\mathcal{U}_s,\ s\in[S],\ 0\leq\alpha_s\leq1,\ \ \ \ \ \ \ \ \ \ \ \ \ \ \ \ \ \ \ \ \ \\
r\sum_{s\in[S]}\alpha_sF\leq B\log_2\left( 1+ \gamma_k^D\right),\ \ \ \ \ \ \ \ \ \ \\
%\forall u\in\mathcal{U}_I^n,\ I\in\mathcal{I}^n,\ n\in&[S]\setminus\{1\}\\
%\sum_{j\in\mathcal{M}_m}R_u^j\leq&\frac{1}{m}\log_2(1+\sum_{j\in\mathcal{M}_m}\gamma_u^j)\\
%\forall\mathcal{M}_m:\mathcal{M}_m\subset I,\ \vert&\mathcal{M}_m\vert=m, m=[n]\setminus\{1\}\\
r(1-\alpha_s)F\leq B\log_2\left( 1+ \gamma_k^s\right),\ \ \ \ \ \ \ \ \ \ \\
%\forall u\in\mathcal{U}_I^n,\ I\in\mathcal{I}^n,\ n\in&[S]\setminus\{1\}\\
r\sum_{j\in\tilde{B}}{(1-\alpha_j)}F\leq B\log_2(1+\sum_{j\in\tilde{B}}\gamma_u^j),\ \ \ \ \ \ \ \ \ \ \\
\text{for all $\tilde{B}\in\mathcal{A}$, $u\in\mathcal{U}_A$, } A\subset [S],\text{ and}\ \ \ \ \ \ \ \ \ \ \ \ \ \ \ \ \ \ \ \ \\
%\forall\mathcal{M}_m:\mathcal{M}_m\subset I,\ \vert&\mathcal{M}_m\vert=m, m=[n]\setminus\{1\}\\
\norm{\textbf{w}_D}_2^2+\sum_{s\in[S]}\norm{\textbf{w}_s}_2^2 \leq P.\ \ \ \ \ \ \ \ \ \ \ \ \ \ \ \ \ 
\end{split}
\end{equation}

\textit{Service Time:} For the above reformulations, the service time $T^*$, is given by $T^*=\frac{F}{B r^*}$ for MMF (\ref{eq:opt_re}), MMF-SIC (\ref{eq:opt_SIC_re}) and $T^*=\frac{1}{r^*}$ for MMF-RS (\ref{eq:opt_RS_re}) Beamforming, where $r^*$s are the optimal values of the objectives in (\ref{eq:opt_re}), (\ref{eq:opt_SIC_re}) and (\ref{eq:opt_RS_re}).

\textit{Solver:} The MMF problems (like $P1$-$P3$) are known to be NP-Hard \cite{KhalajMISO}. However good sub-optimal %points can be obtained using methods like Successive Convex Approximation (SCA), and reformulations as Second Order Cone Problem (SOCP) \
algorithms such as SCA, SOCP are available \cite{KhalajMISO}. %Since the objective of this paper is to bring out the optimal beamforming strategy for queueing, we do not get into details of these reformulations. 
We use the reformulations (\ref{eq:opt_re}), (\ref{eq:opt_SIC_re}) and (\ref{eq:opt_RS_re}) along with Python's, SciPy SLSQP, to directly solve for the beamformers.  We have seen that this gives better performance than SOCP in our simulations\CE.%, directly to solve our optimization problems. Infact, we have seen in our simulations that SLSQP gives better performance than SOCP \cite{KhalajMISO}, for $S=1$.%Infact, we have seen that SLSQP offers better optimal points compared to SOCP, \cite{KhalajMISO}, implemented using CVXPY \cite{cvxpy}, for $S=1$. \CE 

\section{Dual Simple Multicast Queue (DSMQ)}
\label{sec:DSMQ}
Performance of SMQ in MISO with MMF beamforming schemes is severely affected by the presence of bad channel users. This is because (\ref{eq:opt}), (\ref{eq:opt_SIC}) and (\ref{eq:opt_RS}) maximise the min rate, which is controlled by users with bad channel statistics. Under stationarity, by PASTA \cite{Asmussen}, the users with good channel also experience the same mean sojourn times as bad users. Hence the performance of the overall system degrades\CE. %We also observe in \cite{report} that schemes such as loopback, defer etc., \cite{TWC2021}, and power control in time \cite{Arxiv2019}, which was very useful in SISO case is not helpful in our MISO setup with heterogenous channel users\CE. %We also claim that it is not possible to decrease the mean sojourn times of good users while using a single queue. To prove this we give the following theorem:
%\begin{theorem}
%Under IRM traffic, for SMQ, only max-min Fair schemes are delay optimal.
%\end{theorem}
%\begin{proof}
%Let the rate pair $\{R_G,R_B\}$ be the rates of good user and bad user, obtained from an MMF scheme under power constraint $P$ and power allocation $\textbf{w}$. Clearly $R_G\geq R_B$. Hence due to symmetric rate constraints and PASTA, the mean sojourn time $D_{MMF}$ of SMQ is dependent on  $R_B$. Next, we observe that max-min Fair schemes are Pareto efficient \cite{HAL}. Now due to Pareto efficiency, any modification $\textbf{w}'$ to the power allocation to improve rates of good users such that the new rate for good user $R'_G>R_G$, it is necessary that $R'_B<R_B$. Thus, again by symmetric rate constraint and PASTA, for the new mean sojourn time $D'$ obtained by $w'$, we have $D'>D_{MMF}$. 
%\end{proof}
To improve the performance of good channel users while maintaining fairness to bad channel users \CE in multi-antenna case\CE, we modify the SMQ scheme. 

We call the new kind of queue as Dual Simple Multicast Queue (DSMQ). In DSMQ, the requests from the good and the bad channel users are put in two different simple multicast queues, SMQ-G and SMQ-B, respectively. We assume that the BS keeps track of the statistics \CK(e.g., mean/moving average of channel gains\CE) of each user and thus can differentiate between good and bad channel users\CK. When a user channel changes (e.g., due to movement) the BS appropriately changes the queue (SMQ-G/B) for the user (such variations are not in the scope of this work)\CE. %In a practical scenarios these these groups may be time varying (due to user movement, channel non-stationarity etc.). We assume that such variations are not drastic and the BS regroups the users when such changes occur. Study for such variations is not in the scope of this work\CE. 

Further, only one queue is serviced at a time, using all the antennas. Depending on the setting the BS solves $P1$, $P2$ or $P3$ and serves users in first $S$ head of the line files of the queue. We fix a number $C$ and we allow the SMQ-B to be serviced once in every $C$ channel uses or when SMQ-G is empty. If both SMQ-G and SMQ-B are empty the first arrival to the system is served. In all the other conditions, only SMQ-G is serviced.  This way we decouple the QoS (delay) of good channel users from that of bad channel users. 

\CK We will see in Section \ref{sec:sim} that %DSMQ with the 
an appropriate choice of $C$ and beamforming strategy improves the QoS of good channel users without drastically affecting the bad channel users\CE. 

\CE DSMQ can be generalised to the case when there are $G (\geq2)$ groups of users with sufficiently different channel statistics (\CK e.g., 15dB difference in $g_k$s\CE). Here, the number of queues can be increased to, $G$. However, increasing the number of queues reduces the multicast opportunities and perhaps the performance. Thus choosing the number of queues is a tradeoff between fairness and the system performance\CE. 

\CR
\textit{Remarks:} We record here, that the DSMQ, is a culmination of evaluation of multiple schemes that could serve as a potential candidate for providing differentiated QoS (delay). One possible candidate (named ``2Q Simultaneous'', see Section \ref{sec:sim}) is the one that maintains two different queues as DSMQ and serves all the users in head-of-the-line, in both the queues simultaneously, via two streams using one of the MMF beamforming schemes ($P1-P3$). This gives poor performance mainly because the service times become coupled for both the queues, thus bringing down the performance of the good user queue and the overall system. 

We note that  in literature \cite{huang2012stability,queuestable,Qaware,delayperform,largedevqueue}, queueing systems for MIMO/MISO use individual queues for individual users and the transmission/service is performed simultaneously for a subset of users in the head-of-the-line of these queues. These are generalizations of the ``2Q Simultaneous'' queueing system, albeit without merging of redundant requests. We show via simulations, that splitting queues this way (particularly for homogenous system), is detrimental to the performance, since the splitting of queues reduces the multicast opportunities offered by redundant requests for the same files in CCNs.

Further, we also point out that DSMQ is useful only in heterogeneous case. For homogenous case, where we are interested in overall mean sojourn time, we show in Section \ref{sec:sim}, that DSMQ gives poor performance.\\

\CE

%allotted to different users depending on their channel conditions. 

%\subsection{DSMQ with Power Control}
%We can for the improve the performance of the good and bad channel users be performing power control in time. We choose transmit powers $P_G$ and $P_B$, during service of SMQ-G and SMQ-B respectively. We choose the powers such that 
%
%\begin{equation}
%\frac{P_G (C-1) E[S_G]+P_BE[S_B]}{(C-1) E[S_G]+E[S_B]}\leq P
%\end{equation} 
%
%Where $E[S_G]$ and $E[S_B]$, are mean service times of SMQ-G and SMQ-B while using $P_G$ and $P_B$ respectively. Thus $(C,P_G,P_B)$ have to be tuned simultaneously to achieve an optimal QoS. \CE Optimization using DSGD in progress\CE.
\section{\CE Stationarity of SMQ and DSMQ\CE}
\label{sec:station}
\CK Before we proceed it is important to establish the existence of stationarity of the proposed queueing systems. Towards this, we define the state of the queue at a given time $t$, as $X_t=\{(i_1,\mathbb{L}_{i_1}),\cdots,(i_q,\mathbb{L}_{i_q})\}$, where the tuple $(i_j,\mathbb{L}_{i_j})$ represents the $j^{th}$ queue entry of file $i_j$ and $\mathbb{L}_{i_j}$ is the list of users requesting file $i_j$ and $q$ is the queue length. Let, $E[T]$ be the mean service time when all the users request all the $S$ head of the line files (for any given beamformer setting, $P1,P2$ or $P3$).  We assume that $E[T] < \infty$. Let $D_j$ denote the sojourn time of $j^{th}$ request arrival to the queue (SMQ/DSMQ). For DSMQ let $\{X_t^G,X_t^B\}$ be the state of the queue, where $X_t^G$ and $X_t^B$ are defined for the good user queue and the bad user queue in a similar manner. We have the following proposition:

\begin{prop}
\label{prop:station}
Under IRM, if $E[T]<\infty$, then for SMQ and DSMQ, $\{D_j\}$ is an aperiodic regenerative process with finite mean regeneration interval and hence has a unique stationary distribution. Also, starting from any initial distribution, $\{D_j\}$ converges in total variation to the stationary distribution. 
\end{prop} 
%\begin{proof}
%\CE The proof uses Markov chains and regenerative arguments \cite{Asmussen}, to show that the mean regeneration length of the queue state is finite. This is further used to show the existence and uniqueness of the stationary distribution of $\{D_j\}$. See, extended version \cite{report} for a proof\CE.
%\end{proof}
\begin{proof}
%%We will first prove it for SMQ. We first look at discrete time model of the queue. 
Let, $Y_n$, be the state of the queue just after $n^{th}$ departure. Since, there are only $N$, finite number of files in the library, by IRM assumption, $\{Y_n\}$ is a finite state, irreducible discrete time Markov chain (DTMC), \cite{Asmussen}. Now, we show that $\{Y_n\}$ is also aperiodic. To see this, consider the state $Y_n=\{\phi\}$, that is the queue is empty. We note that $P(Y_{n+1}=\{\phi\} \vert Y_n=\{\phi\})>0$, since starting from $Y_n=\{\phi\}$, the event that there is exactly one arrival, has positive probability. Thus the DTMC has a unique stationary distribution. 

Next, consider the delay $D_j$ of the $j^{th}$ arrival to the system just after $Y_n=\{\phi\}$. The epochs, $Y_n=\{\phi\}$ are also regeneration epochs for $\{D_j\}$. Let $E[\tau]$ be the mean regeneration length of $\{Y_n\}$. %%Further, assume that the mean service time/transmission time, $E[T]<\infty$ (see note below for cases when this is true). This bound holds for all arrival rates. 
The mean number of total request arrivals to the BS during these regeneration epochs, defined as $\overline{\eta}$ is bounded by $\lambda E[\tau] E[T]+1$ which is finite. Further $\overline{\eta}$ is also the mean regeneration length of the $\{D_j\}$ process. Therefore, $\{D_j\}$ also has finite mean regeneration length and is aperiodic by the argument given for $\{Y_n\}$. Thus $\{D_j\}$ has a unique stationary distribution. Also, starting from any initial distribution, $\{D_j\}$ converges in total variation to the stationary distribution. 
\end{proof} 

Similarly, we can show that $\{X_t\}$ also has a unique stationary distribution and that starting from any initial distribution, converges in total variation to the stationary distribution. 

We can also show stationarity for DSMQ, in a similar manner by considering the state of the two queue system $\{Y^G_n,Y^B_n\}$ just after the $n^{th}$ departure. %%and proving that  $Y^G_n,Y^B_n=\{\phi,\phi\}$ are regeneration epochs. 
%%This process is also a finite state DTMC which is aperiodic and irreducible. It's regeneration epochs are when $(Y^G_n,Y^B_n)=\{\phi,\phi\}$. The stationarity of the $\{X^G_t,X^B_t\}$ and its mean sojourn time can be proved in a manner similar to that of $\{X_t\}$.
We remark that the stationary distribution itself can be quite complicated, given that the dimensionality of $Y_n$ and $(Y^G_n,Y^B_n)$ will be very large. 

Further, since $E[T]<\infty$ and the queue length is bounded by $N$, the mean sojourn time is upper bounded by $(N+1)E[T]$. This is a unique feature of our queue. 

{The assumption that $E[T]<\infty$, can be satisfied with slight modification to the service of SMQ/DSMQ. %First we lower bound the optimal rate $r^*>r_\epsilon$ where $r_\epsilon$ is a positive quantity close to zero. 
Let, $r^*=1/T^*$, where $T^*$ is the optimal service time obtained after solving (\ref{eq:opt_re}), (\ref{eq:opt_SIC_re}) or (\ref{eq:opt_RS_re}). Further we, fix $r_\epsilon>0$, such that, if  $r^*< r_\epsilon$, we do not transmit for $1/r_\epsilon$ secs.  This time is typically greater than the coherence time after which the channel, $\textbf{H}$ changes. We then solve for the new $\textbf{H}$ and repeat the procedure, say $n$ times, till $r^*\geq r_\epsilon$. Here, $n$ has a geometric distribution with parameter $p$, where $p$ is the probability, $P(r^*\geq r_\epsilon)$.  We choose $r_\epsilon$ such that  $1-p$ is close to zero. Thus $E[T]\leq (1+p)/(p r_\epsilon)<\infty$.}\CE %\textit{The assumption $E[T]<\infty$, is valid in settings where the channels fading coefficients take values in a finite set with positive values, Truncated Rayleigh etc. which can be as close to the complex Gaussian model as required, depending on the discretisation or the truncation\CE. See more details in the extended version \cite{report}}

\CG
\section{Theoretical Approximations of Mean Sojourn Time: SMQ and DSMQ}
\label{sec:smq_theory}
In this section we derive approximate mean delays/sojourn times for SMQ and DSMQ\CG. In \cite{TWC2021} we derived approximate mean sojourn time for the Simple Multicast Queue for single antenna case where the service time distribution was i.i.d, its first and second order statistics were known and also when only head of the line was serviced $(S=1)$. However, in Multi-Antenna case the service time distribution is dependent on
\begin{enumerate}
\item the distribution of channel gains, of different users, 
\item the solutions of the optimization problems $P1-P3$ and 
\item the stationary user distribution in the queue\CG. 
\end{enumerate}
Among these the distribution of channel gains is known (from assumptions). A closed form solution to the optimization problems $P1-P3$ is not available as the problem is known to be NP Hard\CR. Further, the stationary user distribution itself is dependent on service time distribution\CG. Thus, although the service time distribution in principle can be obtained given the channel gain distribution and the user distribution, its closed form is not available. But it can be easily obtained in simulation. We denote the first and second moment of service time as $\overline{T_\Sigma(\pi)}$ and $\overline{T^2_\Sigma(\pi)}$ for a given channel gain distribution $\Sigma$ and user distribution $\pi$.  %But, the user distribution of the file requests in the queue itself is unknown and complicated as explained in section \ref{sec:station}. Thus i
In this section we derive approximations for the user distribution also, following that we propose a method for getting the mean sojourn times. 

As a first step we approximate the mean sojourn time of SMQ, when the services have a known general and independent distribution ($GI$ services) as in \cite{TWC2021} and extend it to $S\geq2$. This approximation will be further generalized to the multi-antenna case.

\subsection{Approximation of Mean Sojourn Time of SMQ for $GI$ services}
\label{sec:smq_mst}
In SMQ for $S=1$, we first typify the requests into two types, Type 1 (T1) and Type 2 (T2). Among the overall requests arriving at the Base Station queue, T1 requests correspond to the first request for a file that is either unavailable in the queue (or) is in service. Thus when a T1 request arrives at the queue it is put at the tail of the queue and the subsequent requests (T2) for the file are merged with this entry, as mentioned in section \ref{sec:system_model}. Let $\lambda_i'$ correspond to the rate of arrival of file $i$ in T1. Now, let $D_{1i}$ represent the delay experienced by the T1 arrival of file $i$. By PASTA (\cite{Asmussen}\CG), the mean number of total arrivals to the queue till this request is served is given as $1+\lambda_i E[D_{1i}]$, where $\lambda_i=\sum_{j=1}^K \lambda_{ij}$. Thus, 

\begin{equation}
\label{eq:eff_lam}
\lambda_i'=\frac{\lambda_i}{1+\lambda_i E[D_{1i}]}.
\end{equation}

Further, the rate of total T1 traffic for all the files is given as $\lambda'=\sum_{i=1}^N \lambda_i'$. We further assume $E[D_{1i}] = d$, for some $d$ and all $i$. In \cite{TWC2021}, we showed that this approximation is quite accurate and have verified that it holds even for varying file sizes. Finally we approximate the mean delay of T1 traffic using P-K formula for M/G/1 queues \cite{Asmussen} as 
\begin{equation}
\label{eq:mg1}
d= \frac{\rho_d}{1-\rho_d}  \left(\frac{E[T^2]}{2 E[T]}\right),
\end{equation}
where $\rho_d=E[T]\lambda' =E[T]\sum_{i=1}^N \frac{\lambda_i}{1+\lambda_i d}$ and $E[T]$ and $E[T^2]$ are first and second order means of i.i.d service times. The mean delay of T1 traffic can be calculated by recursively solving (\ref{eq:mg1}), starting from any $d$. Further, if $d^*$ is the T1 delay obtained from the recursion, the delay experienced by T2 traffic has mean $d^*/2$. This is because the overall traffic is Poisson and hence the T2 arrivals have uniform distribution in the duration $d^*$. Thus the overall mean sojourn time is given as 
\begin{equation}
\label{eq:mst}
\overline{D}=\frac{\lambda'}{\lambda} d^* + \frac{\lambda-\lambda'}{\lambda} \frac{d^*}{2} + E[T].
\end{equation} 

For $S>1$, since more than $1$ file is being served in a service the utility $\rho_d$ reduces $S$ times. Thus, \CG we make another approximation \CG in (\ref{eq:mg1}) and replace $\rho_d$ by $\rho_d/S$ and we get 
\begin{equation}
\label{eq:mg1_S}
d= \frac{\rho_d}{S-\rho_d}  \left(\frac{E[T^2]}{2 E[T]}\right)=f(d,S,E[T],E[T^2]).
\end{equation}
Again, the mean sojourn time can be derived using (\ref{eq:mst})\CG. It is worth noting that the SMQ may not always have queue length greater than $S$, particularly at low arrival rates. However, (\ref{eq:mg1_S}) assumes that the queue length is always at least $S$. This is a source of error which could be large when $S$ is large. Thus, (\ref{eq:mg1_S}) gives a lower bound on Type 1 delay in such cases, since $\rho_d$ is reduced $S$ times. %We will also see that, in MISO case, since service time depends on the parameter $S$, the equivalent equation for MISO SMQ (\ref{eq:smq_fp}), gives an upper bound at lower arrival rates (seen via simulations). 

Following comments are in order:
\begin{itemize}
\item The intuition behind the assumption $E[D_{1i}] = d$ is that the T1 traffic is a thinned superposition of the original \CG Poisson processes \cite{kallenberg}. With appropriate conditions (Theorem 16.18, \cite{kallenberg})\CG, the thinned, superposed process tends to a Poisson process. Thus if the T1 process can be approximated as Poisson, then PASTA holds and hence the assumption. Indeed we have seen in simulations that our assumption is true for a wide range of system parameters $\gamma$, $\lambda_{ij}$s, $N$, $K$, etc.
\item Further, the M/G/1 approximation also holds true due to the above reason.
\item Similar to \cite{TWC2021}, via a qualitatively typical example we show that the fixed point equation in (\ref{eq:mg1_S}) has a unique solution for any service time. We have verified that the curves for different $\lambda,\ K,\ S,\ N,\ \gamma$s are typically similar to those presented in this section. We consider a system with $\lambda=40$, $K=40$, $S=1$, $N=100$ and the file popularity has Zipf distribution with $\gamma=1$.  Further lets consider two pairs of first and second moments of service times $(E[T],E[T^2])=$$(S_1,S_1^2)$ and $(S_2,S_2^2)$ where $S_1=0.15$ and $S_2=0.2$. From Figure \ref{fig:unique_fd} we see that for both the pairs of service time moments have unique solutions for $d=f(d,1,E[T],E[T^2])$. It is also worth noting that the unique solutions lie in the region (right of black vertical lines in figure \ref{fig:unique_fd} for each pair of service moments) where $\rho_d/S <1$\CG. We also see from Figure \ref{fig:unique_fd}\CG, that $d^*$ increases, with increase in moments. This is typically the case for all system parameters. In fact, we have seen that $d^*$ is non-decreasing in both $E[T],\ E[T^2]$ for different sets of system parameters. 
\item To see the stability of the queue analytically, in (\ref{eq:mg1_S}), if $\lambda \rightarrow \infty$, then $f \rightarrow \frac{{E}[T^2] N/ d}{2(S-{E}[T] N/d)} < \infty $. Thus solving the quadratic equation from (\ref{eq:mg1_S}), arising in the limit, we get
\begin{equation}
d=\frac{2E[T]N\pm \sqrt{4(E[T]N)^2-8E[T^2]NS}}{4S}.
\end{equation} 
Now, ignoring the negative (irrelevant) root, for large $N$, we get $d\sim E[T]N/S$. Thus Type-1 delay (and hence the mean sojourn time in (\ref{eq:mst})) is upper bounded by $E[T]N/S$ for all $\lambda$s. Thus our queue is stable for all arrival rates, a feature unique to our queue.
%\item In \cite{TWC2021} we have also shown that for high traffic the mean T1 delay is analytically bounded by $N E[T]$. 
\end{itemize} 

\begin{figure}[!h]
\centering
\includegraphics[height=4cm,width=8cm,trim={1cm 9.5cm 1cm 9.5cm},clip]{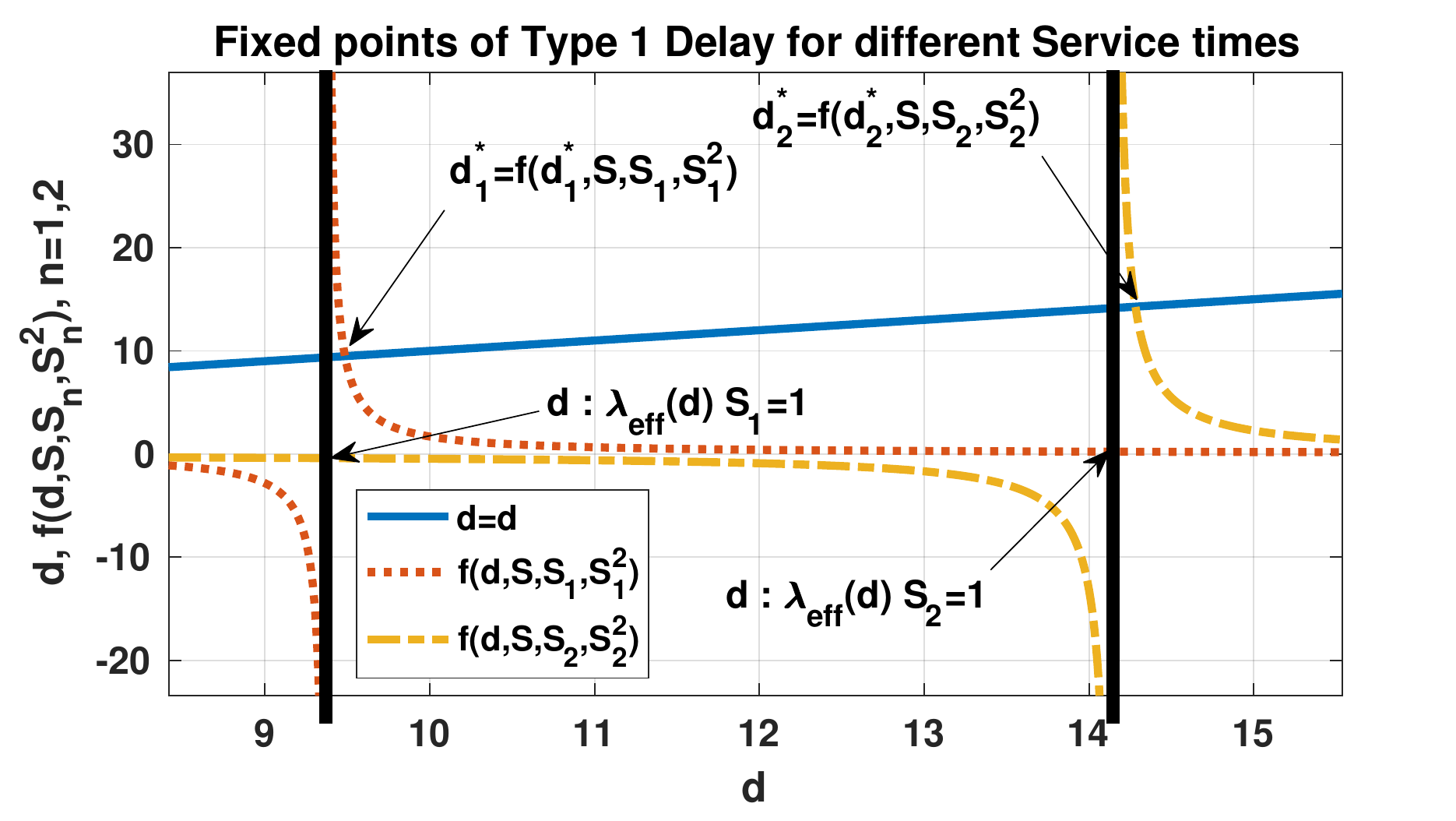}
\caption{Uniqueness of solution to the fixed point (\ref{eq:mg1_S}): $K=40,\ \lambda=40,\ N=100,\ \gamma=1$, $S_1=.15$ secs, $S_2=.2$ secs}
\label{fig:unique_fd}
\end{figure}

Now, we proceed with approximation of distribution of users for different file requests, under stationarity, in the queue.

\subsection{Approximate User Distribution for SMQ}
\label{sec:smq_approx}
\CG We are concerned with the distribution of users requesting first $S$ head of the line files (service times depend only on this). We first derive the user distribution for the case $S=1$ and then generalize it to $S\geq 1$. Towards this note that the probability that file $i\in[N]$ is in the queue is given as 
\begin{equation}
\label{eq:file_d}
P(i \in Q)=\frac{\lambda_i'}{\lambda'}.
\end{equation}

Further, the probability that a T1 request for file $i$ is made by user $j$ is given by
\begin{equation}
P_{ij}=\frac{\lambda_{ij}}{\lambda_i}.
\end{equation}

Let $\mathbb{L}_{ij}$ represent the set of users requesting file $i$ when the T1 request was made by user $j$. Now, just before the service of file $i$, the $\mathbb{L}_{ij}$ would contain all the users who have made atleast one request within time $D_{1i}$. We approximate the probability that user $k\neq j$ makes atleast one request within time $D_{1i}$ as 
\begin{equation}
\label{eq:user_request}
q^k_{ij}=P(k\in \mathbb{L}_{ij})=1-e^{-\lambda_{ik}E[D_{1i}]}=1-e^{-\lambda_{ik}d}.
\end{equation}
In the above expression we use the approximation in the previous section, $E[D_{1i}]=d$ for all $i$. Now, let $U_1\in\{0,1\}^K$ represent a $K$ length vector where $k^{th}$ element is $1$, \textit{if the $k^{th}$ user has requested the file in head of the line} and $0$ otherwise. We can now approximate the conditional user distribution for head of the line entry given the requested file is $i_1$ as
\begin{multline}
\label{eq:cond_ud_1}
P(U_1=[u_1^1,\cdots,u_1^K]\vert i_1)=\\\prod_{j=[K]}P_{i_1j} \prod_{k\neq j,k=[K]}(u_1^k q^k_{i_1j}+(1-u_1^k)(1-q^k_{i_1j}) ).
\end{multline}   
Combining (\ref{eq:file_d}) and (\ref{eq:cond_ud_1}) we get the head of the line user distribution as 

\begin{multline}
\label{eq:ud_1}
\pi^d_{1}=P(U_1=[u_1^1,\cdots,u_1^K])=\\ \sum_{i_1\in[N]}P(i_1 \in Q) P(U_1=[u_1^1,\cdots,u_1^K]\vert i_1)
\end{multline}

\textit{$S\geq 1$ case:} Proceeding as for $S=1$, the probability that first $S$ entries have files $i_s,s=[S]$, $i_l\neq i_q$, for $l\neq q$ is given as 
 
\begin{equation}
\label{eq:file_d_s}
P(i_s \in Q, s=[S])=\frac{\lambda_{i_1}'}{\lambda'}\prod_{s=[S-1]}{\frac{\lambda_{i_{s+1}}'}{\lambda'-\sum_{p=[s]}{\lambda_{i_p}'}}},
\end{equation}
where $Q$ is the set of files in SMQ. 

Note that this is an approximation of the actual distribution since the Type 1 delay will actually be less than $d$ for queue entries (file request positions) other than the head of the line entry. However, we calculate all $\lambda'_{i_s}$, $\forall i_s \in Q, s=[S]$ using the same $d$. The error due to this approximation may increase for large $S$. Nevertheless, we will see that for the values of interest of the parameter $S$ ($=1$,$2$ and $3$), the approximations work well and give low error. 

Now, let $U_s\in\{0,1\}^K$ represent a $K$ length vector where $k^{th}$ element is $1$, \textit{if the $k^{th}$ user has requested the file in $s^{th}$ entry of the queue} and $0$ otherwise. We can now define the conditional user distribution in the first $S$ head of the line entries given the files in positions $1,\cdots,S$ are $i_1,\cdots,i_S$ respectively, as
\begin{multline}
\label{eq:cond_ud}
P(U_s=[u_s^1,\cdots,u_s^K], s=[S]\vert i_s, s=[S])=\\\prod_{s=[S],j=[K]}P_{i_sj} \prod_{k\neq j,k=[K]}(u_s^k q^k_{i_sj}+(1-u_s^k)(1-q^k_{i_sj}) ).
\end{multline}   
Combining (\ref{eq:file_d_s}) and (\ref{eq:cond_ud}) we get the user distribution in the $S$ head of the line entries as 

\begin{multline}
\label{eq:ud}
\pi^d_{S}=P(U_s=[u_s^1,\cdots,u_s^K], s=[S])=\\ \sum_{i_1\neq\cdots\neq i_S\in[N]^S}P(i_s \in Q, s=[S])\times\\ P(U_s=[u_s^1,\cdots,u_s^K], s=[S]\vert i_1,\cdots,i_S)
\end{multline}

\CG The superscript $d$ is added to emphasis the dependence of the user distribution on the T1 delay $d$. Thus given the parameter $S$, the user distribution $\pi^d_{S}$ and channel gain parameter $\Sigma$ we can now calculate the first and second moments of service time $\overline{T_\Sigma(\pi^d_{S})}$ and $\overline{T^2_\Sigma(\pi^d_{S})}$. Using the moments of service time along with the M/G/1 approximation in the previous section, enables calculation of the Mean Sojourn Time of SMQ for the MISO System. For simplicity of notation we write $\overline{T_\Sigma(\pi^d_{S})}$ and $\overline{T^2_\Sigma(\pi^d_{S})}$ as $T(d)$ and $T^2(d)$ respectively.

\subsection{Theoretical Mean Sojourn Time for SMQ MISO}
\label{eq:theory_algo}
Having derived $\pi^d_{S}$, we are now equipped to formulate the algorithm to derive the mean sojourn times of SMQ for multi-antenna case. Towards this we first rewrite, (\ref{eq:mg1_S}) as 
\begin{equation}
\label{eq:smq_fp}
d'=f(d',S,T(d),T^2(d))= \frac{T(d)\lambda'(d')}{S-T(d)\lambda'(d')}  (\frac{T^2(d)}{2 T(d)}),
\end{equation}
Notice that the fixed point is over the variable $d'$ and also (\ref{eq:smq_fp}) brings in the dependence of service times on $d$. Though the ideal fixed point equation should be with $d'$ replacing $d$, we will see that keeping $d$ and $d'$ as different variables is useful in iterative algorithm formulation. Also, note that the M/G/1 approximation (\ref{eq:smq_fp}) subsumes that the service times are $i.i.d$ which is not true in reality. However, we will see that this approximation works well nevertheless\CG. We will also see that unlike (\ref{eq:mg1_S}), in the MISO system, since service time depends on the parameter $S$, the equivalent equation (\ref{eq:smq_fp}), gives an upper bound at lower arrival rates (seen via simulations)\CG. The algorithm to calculate Mean sojourn time is described in Algorithm \ref{algo:mst_miso_smq}. This algorithm works as follows:

\begin{itemize}
\item For any given $d$ the moments $T(d)$ and $T^2(d)$ are bounded (the boundedness arguments are similar to the $E[T]<\infty$ argument in Section \ref{sec:station}). 
\item Now, for a given $d$ we get $T(d)=t1$ and $T^2(d)=t2$ as follows: $M$ samples of service times are obtained using $P_X=P1,P2$ or $P3$ on $M$ realisations of user distribution and the corresponding channel gains using $\pi^d_{S}$, $\Sigma$ as shown in Algorithm \ref{algo:mst_miso_smq}. The $t1$ and $t2$ are then obtained by taking sample average of $M$ service times and square of service times.
\item Now iterating the fixed point equation (\ref{eq:smq_fp}) with $t1$ and $t2$, we get a unique $d^*$ such that $d^*=f(d^*,S,t1, t2)$.
\item We replace $d\leftarrow d^*$ and repeat the above steps till $\vert d_{n}-d_{n-1} \vert < \epsilon$ for a given $\epsilon$, where $d_n$ is the value of $d$ in the $n^{th}$ iteration.  
\end{itemize}

\begin{algorithm}[h!]
\CG
\SetAlgoLined
\caption{Mean Sojourn Time for MISO SMQ with MMF, MMF-SIC and MMF-RS beamformers}
\label{algo:mst_miso_smq}
\KwIn{$L$, $K$, $N$, $S$, $P$, $\lambda$, $\gamma$, $\Sigma$, $\epsilon$: error, $M$: Samples}
Initialize: $\overline{T}=1,\overline{T^2}=1$, $d_{prev}=0,d=2\epsilon$\\
$P_X=P1,P2$ or $P3$\\
\While{$\vert d-{d_{prev}}\vert>\epsilon$} 
{
$d_{prev}\leftarrow d$\\
$d_{old}=0$\\
\While{$\vert d-{d_{old}}\vert>\epsilon$}
{
$d_{old}\leftarrow d$,\\ $d\leftarrow f(d,S,\overline{T}, \overline{T^2})$
}
Use (\ref{eq:ud}) and calculate $\pi^d_S$ for the new $d$\\ 
\For{$m=1$ \KwTo $M$}
{
Sample $U_s,s\in[S]$ using $\pi^d_S$\\
Sample $\textbf{H}$ (Channel Gains) using $\Sigma$\\

Solve $P_X$ for $U_s,s\in[S],\ \textbf{H}$ and obtain the service time $T^*$\\
$T_m\leftarrow T^*$
}
${T(d)}\leftarrow \frac{\sum_{m\in[M]}{T_m}}{M}$, ${T^2(d)}\leftarrow \frac{\sum_{m\in[M]}{T_m^2}}{M}$\\

$\overline{T}\leftarrow {T(d)}$, $\overline{T^2}\leftarrow {T^2(d)}$ 
}

${d}^{*}\leftarrow d$\\
/*Calculate Mean Sojourn Time*/
$\overline{D}=\frac{\lambda'}{\lambda} d^* + \frac{\lambda-\lambda'}{\lambda} \frac{d^*}{2} +\overline{T} $\\
\KwOut{$\overline{D}$: Mean Sojourn Time}
\end{algorithm}

%We will see that the mean sojourn time calculated by the algorithm is very close to the simulation for wide range of system parameters. We make the following comments on the Algorithm \ref{algo:mst_miso_smq}:
We have seen that, starting from any $d$ Algorithm \ref{algo:mst_miso_smq} converges to a unique value of Type 1 delay and hence to a unique mean sojourn time. To show this we continue with our typical example in Figure \ref{fig:unique_fd} and provide some heuristics for the convergence of Algorithm \ref{algo:mst_miso_smq} to a unique fixed point of (\ref{eq:smq_fp}). We fix $L=16$, $g=1$ and $P_X=P1$. In general, for any system parameters the plots and the following heuristics are similar. %First we show the existence of the fixed point and then show the uniqueness via numerical evaluation.

\subsection{Heuristics for Convergence of Algortihm \ref{algo:mst_miso_smq}}
\label{sec:heuristics}
We begin by noting that as $d$ increases the number of users in the $S$ entries starting from head of the line, increases. This can be easily seen from (\ref{eq:user_request}) and (\ref{eq:ud}). Thus, naturally the service time increases with $d$ as the number of users in the system increase for all $P_X=P1,P2$ and $P3$. Further for very large $d$, we see that the user distribution is such that all the users are present in all the $S$ entries starting from head of the line. However, we note that the service moments remain bounded (from arguments similar to that of $E[T]<\infty$ in Section \ref{sec:station}).

Further, we have seen that the equation in (\ref{eq:mg1_S}) has a unique fixed point $d^*$ for a given set of values of service moments. Similarly, (\ref{eq:smq_fp}) also has a unique fixed point (here the fixed point equation is in variable $d'$) for each value of $d$. Thus in each iteration, for a given $d$ Algorithm 1 gives a unique ${T}(d)$, ${T}^2(d)$ and $d^*$, where $d^*$ is the fixed point of the equation $d'=f(d',S,T(d),T^2(d))$. In Figure \ref{fig:SMQ_MISO_Theory}, we see that $f(d^*,S,T(d),T^2(d))$ is monotonically increasing with $d$. This is again because of monotonic increase in service time moments with $d$, because of which $d^*$ is non-decreasing as seen in Section \ref{sec:smq_mst}. Further, since Type-I delay $d^*$ and $T(d),T^2(d)$ are bounded, $f(d^*,S,T(d),T^2(d))$ is also bounded for all values of $d$. This can be seen in Figure \ref{fig:SMQ_MISO_Theory} where $f$ is bounded and non-decreasing. 

Now we note that, $T(d)$ and $T^2(d)$ are strictly positive even for $d=0$ since $P_X=P1,P2,P3$ give strictly positive service times even with one user in each entry (the distributions (\ref{eq:user_request}) and (\ref{eq:ud}) are such that there is atleast one user in each entry even for $d=0$). Thus the fixed point equation (\ref{eq:smq_fp}) gives a non-zero positive $d^*$ even for $d=0$. Further assuming continuity (proof of which is beyond the scope of analysis since continuity of $f$ depends on continuity of service moments $T(d)$ and $T^2(d)$ which is further dependent on the NP-Hard problems $P1,\ P2$ or $P3$) the straight line $d$ intersects $f$ at at least one point. We have seen from evaluations of Algorithm 1 for different system parameters that this point is unique.  
%Now, we can say that the line $d'=d$ intersects $d'=f(d^*,S,T(d),T^2(d))$ at atleast one point if $f(d^*,S,T(d),T^2(d))>0$ for $d$ near $0$. To see this let $d=\epsilon$, where $\epsilon>0$, is arbitrarily close to zero. Now, $f(d^*,S,T(\epsilon),T^2(\epsilon))= \frac{T^2(\epsilon)\lambda'(d^*)}{2 (S-\lambda'(d^*)T(\epsilon))} $, which is strictly positive. This is because the service moments are positive for any $\epsilon>0$ and the $d^*$ choosen by (\ref{eq:smq_fp}) is unique and such that $\lambda'(d^*)T(\epsilon)<S$. Thus there exists atleast one fixed point for Algorithm 1.    %  %We note that $d^*,f(d*,.,.,)$ also increase with increase in $d$. However since the Type 1 delay of SMQ is bounded, the monotonic $d^*$ and $f(d*,.,.,)$ converges to a constant value. Thus    

From Figure \ref{fig:SMQ_MISO_Theory} we see that for a typical case, there exists a unique fixed point which the algorithm achieves, starting from any $d$\CG. For every $d$, the algorithm computes the service moments and the corresponding $d^*$ using (\ref{eq:smq_fp}). The algorithm then proceeds iteratively by replacing the $d$ with the new $d^*$ as explained in Section \ref{eq:theory_algo}\CG. The dashed arrows in Figure \ref{fig:SMQ_MISO_Theory} show different paths, the algorithm takes starting from different $d$, to reach this unique fixed point.

\begin{figure}[!h]
\centering
\includegraphics[height=4cm,width=8cm,trim={1cm 9.5cm 1cm 9.5cm},clip]{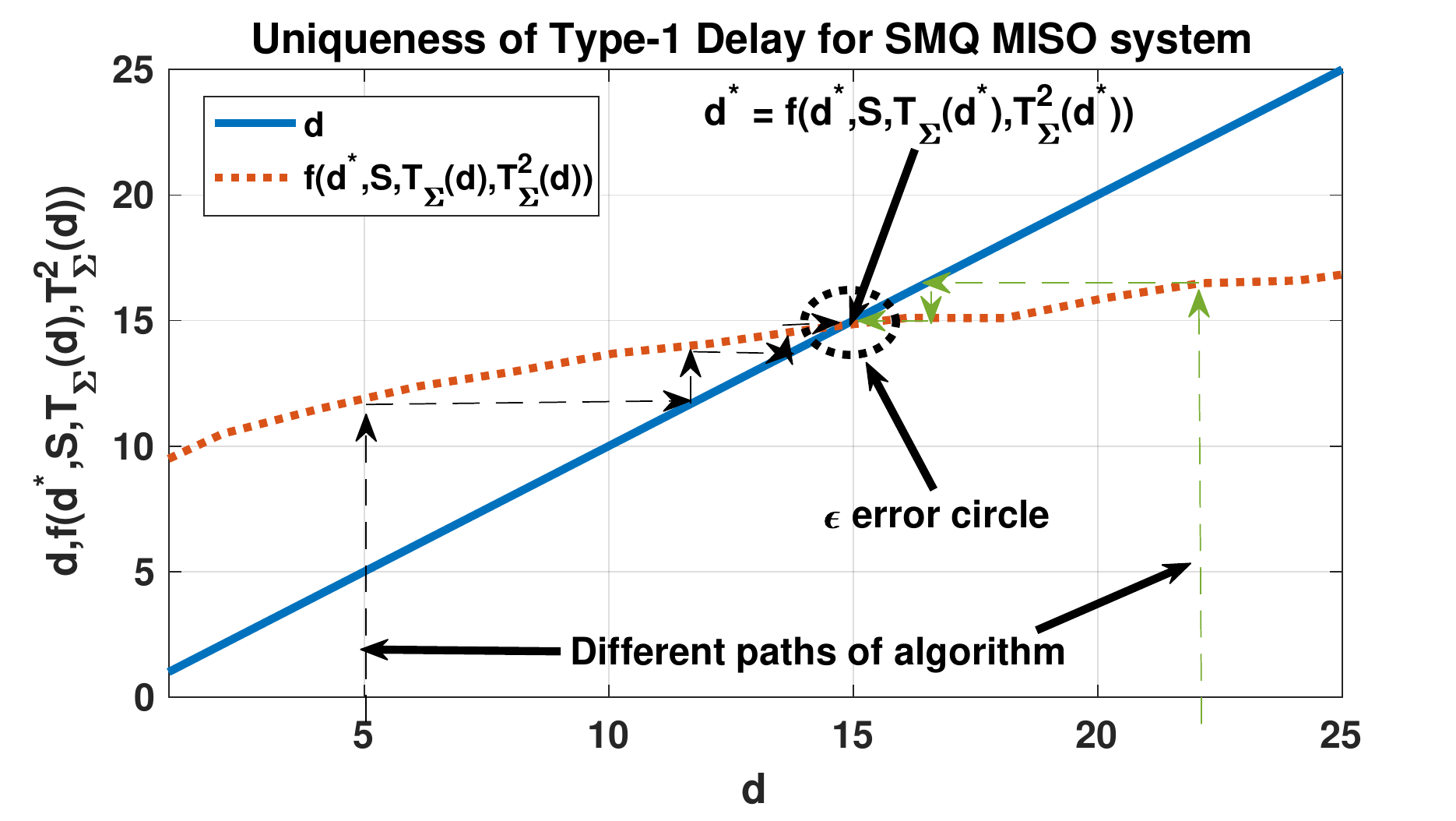}
\caption{Uniqueness of Mean Sojourn Time for SMQ MISO case: $L=16,\ K=40,\ \lambda=40,\ g=1,\ P=10,\ N=100,\ N_0=1,\ \gamma=1$, }
\label{fig:SMQ_MISO_Theory}
\end{figure}

\subsection{Evaluation of Algorithm \ref{algo:mst_miso_smq}:}
\label{sec:eval_algo1}
In this section we evaluate the accuracy of the approximation of mean sojourn time provided by Algorithm \ref{algo:mst_miso_smq}. Towards this we look at a system with $L=16$, $P=10$, $N_0=1$, $N=100$, $F=100Mb$, $B=100MHz$ and the file popularity follows Zipf distribution with parameter $\gamma=1$. The channels are modelled as complex Gaussian flat fading channels with $g=1$. To evaluate the goodness of fit of our approximation we vary the system parameters as $\lambda=10,20$ and $40$ representing low, medium and high loads, $K=4,10,20,30,40$ and $S=1,2$. We compare the mean sojourn times obtained via full system simulation with the ones obtained via Algorithm \ref{algo:mst_miso_smq}. The simulations are run for $10000$ services of the queue at the BS, to let the queues reach stationarity. The mean sojourn times in simulations are calculated using sample average of sojourn times seen during the simulation. To keep the service times bounded we fix $r_\epsilon=0.01$, (see Section \ref{sec:station}). Th algorithm parameters are fixed as follows $P_X=P1$, $\epsilon=0.1$ and $M=500$. The comparison here is shown for MMF beamforming (P1), other beamforming cases (P2 and P3) also have similar performance. 

From Figure \ref{fig:Th_Sim_SMQ} we see that Algorithm 1 (Theory) gives a mean sojourn time with less than $10\%$ error for cases $\lambda=20$ and $40$ (medium and high load) and $<15\%$ for the case $\lambda=10$ (low load case). Slight increase in error for $\lambda=10$, is because of less merging opportunities in low load condition. Hence the system (considering Type-1 traffic) deviates from the M/G/1 like system. We point out that in a system design, medium and high load conditions are more critical where our algorithm gives high accuracies. Thus Algorithm 1 gives a quick method for system designers to evaluate the performance of the system without using complicated full system simulations.     

\begin{figure}[!h]
\centering
\includegraphics[height=6cm,width=8cm,trim={1cm 7cm 1cm 7cm},clip]{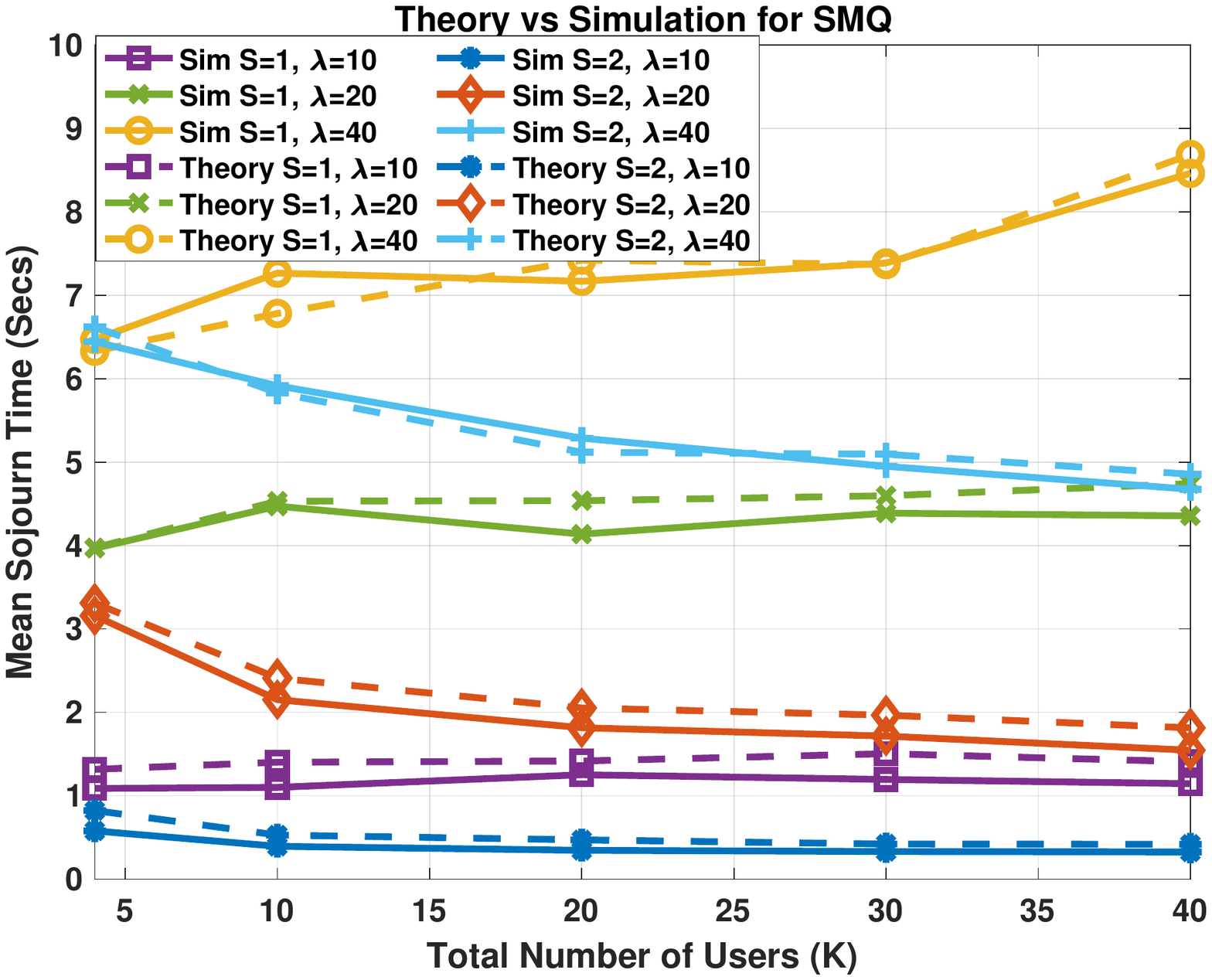}
\caption{Comparison of Theory vs Simulation for SMQ: $L=16$, $P=10$, $N_0=1$, $N=100$, $F=100Mb$, $B=100MHz$, $g=1$}
\label{fig:Th_Sim_SMQ}
\end{figure}

Now we proceed to derive approximate mean sojourn time for DSMQ. First we derive the mean sojourn time for general independent service times.  

\subsection{Approximation of Mean Sojourn Time of DSMQ for $GI$ Services}
\CG We begin by noting that DSMQ is a Polling system \cite{takagibook}, with zero switching time for the server, with E-limited service discipline \cite{fhurmann}.

There are several works which use Pseudo Conservation Laws \cite{boxma,takagi,fhurmann} to derive approximate waiting times in the queues of polling systems. However, these approaches are not conducive to solving the fixed point equations (similar to (\ref{eq:mg1_S}) and (\ref{eq:smq_fp})) arising due to merging of traffic and hence cannot be used directly.  We propose an alternate approach where we consider SMQ-G and SMQ-B as independent M/G/1 queues with some modifications to the service times\CG. %In section \ref{sec:sim} we will see that our approach gives an accuracy of $>90\%$ for Medium and High arrival rates and $>85\%$ for low arrival rates\CE. 
Towards this we first derive independent M/G/1 approximations for DSMQ for both the queues when the service times are i.i.d. 

\CG Let $d_1$ and $d_2$ be the Type 1 delays of SMQ-G and SMQ-B of DSMQ. Here, we are again making an approximation that for each queue the corresponding file requests see the same mean Type 1 delays, for reasons similar to SMQ.   

Let $\overline{T_1}$ and $\overline{T_2}$ be the mean service times and let $\overline{T^2_1}$ and $\overline{T^2_2}$ be the second moments of service times of SMQ-G and SMQ-B, respectively. Now, SMQ-G can be approximated as an independent M/G/1 queue with the first and second moments of service times defined as
\begin{multline}
\label{eq:mix_serve_G}
\overline {T_G}=\overline{T_1}+\overline{T_2}/(C-1),\ 
\overline {T^2_G}=\overline{T^2_1}+\overline{T^2_2}/(C-1)
\end{multline}
and effective arrival rate as in Section \ref{sec:smq_approx} considering only rates $\Lambda_G=\{\lambda_{ij}$, $i\in[N]$ and $j\in \mathcal{K}_G\}$, where $\mathcal{K}_G$ is the set of good channel users. Similarly, SMQ-B can be approximated as an M/G/1 queue with the first and second order mean service times given as
\begin{multline}
\label{eq:mix_serve_B}
\overline {T_B}=\overline{T_1}(C-1)+\overline{T_2},\ 
\overline {T^2_B}=\overline{T^2_1}(C-1)+\overline{T^2_2}.
\end{multline}
Also, the effective arrival rate as in section \ref{sec:smq_approx} considering only rates $\Lambda_B=\{\lambda_{ij}$, $i\in[N]$ and $j\in \mathcal{K}_B\}$, where $\mathcal{K}_B$ is the set of bad channel users. Now similar to SMQ, we can approximate mean T1 delay of SMQ-G using the fixed point equation 
\begin{equation}
\label{eq:mg1_G}
d_1= f_G(d_1,C,S,\overline{T_G},\overline{T^2_G})=\frac{\rho_{d_1}}{S-\rho_{d_1}}  (\frac{\overline{T^2_G}}{2 \overline{T_G}}),
\end{equation}
where $\rho_{d_1}=\overline{T_G}\lambda_G'(d_1)$, and $\lambda_G'=\sum_{i=1}^N \frac{\sum_{j\in \mathcal{K}_G\}}\lambda_{ij}}{1+ d_1{\sum_{j\in \mathcal{K}_G\}}\lambda_{ij}}}$.

Also, for SMQ-B the mean T1 delay is obtained using the fixed point equation 
\begin{equation}
\label{eq:mg1_B}
d_2= f_B(d_2,C,S,\overline{T_B},\overline{T^2_B})=\frac{\rho_{d_2}}{S-\rho_{d_2}}  (\frac{\overline{T^2_B}}{2 \overline{T_B}}),
\end{equation}
where $\rho_{d_2}=\overline{T_B}\lambda_B'$, and $\lambda_B'(d_2)=\sum_{i=1}^N \frac{\sum_{j\in \mathcal{K}_B}\lambda_{ij}}{1+ d_2{\sum_{j\in \mathcal{K}_B}\lambda_{ij}}}$. Thus T1 delays, $d_1$ and $d_2$  for both SMQ-G and SMQ-B can be calculated using recursive equations (\ref{eq:mg1_G}) and (\ref{eq:mg1_B}) independently. The uniqueness of $d_1$ and $d_2$ holds true even in this case as both the queues are treated independently. Further, the user distributions $\pi^{d_1}_S$ and $\pi^{d_2}_S$ for SMQ-G and SMQ-B can be derived in the same manner as in Section \ref{sec:smq_approx}, with rates $\Lambda_G$ and $\Lambda_B$ respectively. 

\subsection{Theoretical Mean Sojourn Time for DSMQ MISO}
\CG We can extend Algorithm \ref{algo:mst_miso_smq} to derive mean sojourn time for the DSMQ MISO system. The DSMQ can again be considered as a Polling system with $\Lambda_G$ arrivals to SMQ-G and $\Lambda_B$ arrivals to SMQ-B. Similar to $T(d)$ and $T^2(d)$ in the SMQ case, the first and second moments of the service times of SMQ-G and SMQ-B are represented as $T_n(d_n), T_n^2(d_n)$, $n=1,2$, respectively.  The moments $T_n(d_n), T_n^2(d_n)$, $n=1,2$ are dependent on $\Sigma$, $\pi^{d_1}_S$ and $\pi^{d_2}_S$. Similar to the SMQ case we replace $\overline{T_n}$ and $\overline{T^2_n}$, $n=1,2$, in equations (\ref{eq:mix_serve_G}) and (\ref{eq:mix_serve_B}) with $T_n(d_n), T_n^2(d_n)$, $n=1,2$ respectively and rewrite (\ref{eq:mg1_G})  and (\ref{eq:mg1_B}) as follows:

\begin{equation}
\label{eq:mg1_miso_G}
d_1'=f_G(d_1',C,S,\overline{T_G},\overline{T^2_G})=\frac{\lambda_G(d_1')\overline{T_G}}{S-\lambda_G(d_1')\overline{T_G}}  (\frac{\overline{T^2_G}}{2 \overline{T_G}}).
\end{equation}
Similarly for SMQ-B the mean T1 delay is obtained using the fixed point equation
\begin{equation}
\label{eq:mg1_miso_B}
d_2'= f_B(d_2',C,S,\overline{T_B},\overline{T^2_B})=\frac{\lambda_B(d_2')\overline{T_B}}{S-\lambda_B(d_2')\overline{T_B}}  (\frac{\overline{T^2_B}}{2 \overline{T_B}}).
\end{equation}
Note that unlike SMQ case, the $\overline{T_G}$ and $\overline{T_B}$ are both dependent on both $d_1$ and $d_2$ as:
\begin{equation}
\label{eq:mix_serve}
\begin{split}
\overline {T_G}&=T_1(d_1)+T_2(d_2)/(C-1),\\ 
\overline {T^2_G}&=T^2_1(d_1)+T^2_2(d_2)/(C-1),\\
\overline {T_B}&=T_1(d_1)(C-1)+T_2(d_2),\\
\overline {T^2_B}&=T^2_1(d_1)(C-1)+T^2_2(d_2).
\end{split}
\end{equation}
However, similar to the SMQ case, the fixed point equation is over the variable $d_1'$ alone for (\ref{eq:mg1_miso_G}) and $d_2'$ alone for (\ref{eq:mg1_miso_B}).

The Algorithm \ref{algo:mst_miso_dsmq} proceeds in a manner similar to Algorithm \ref{algo:mst_miso_smq}, except for the coupling in service times coming from equations (\ref{eq:mix_serve_G}) and (\ref{eq:mix_serve_B}). The pseudo-code to calculate the mean sojourn time is given in Algorithm \ref{algo:mst_miso_dsmq}. 

\CG The uniqueness of the fixed points of Algorithm \ref{algo:mst_miso_dsmq} is seen by heuristics, similar to that of SMQ case (Section \ref{sec:heuristics}). This is because the fixed point iterations run in a decoupled manner even though $\overline{T_G}$ and $\overline{T_B}$ are coupled\CG. 
%\begin{itemize}
%\item ---------- Need to talk about convergence here as well -----------
%\end{itemize}

\begin{algorithm}[h!]
\CG
\SetAlgoLined
\caption{Mean Sojourn Times for MISO DSMQ with MMF, MMF-SIC and MMF-RS beamformers}
\label{algo:mst_miso_dsmq}
\KwIn{$L$, $K$, $N$, $S$, $P$, $\Lambda_G$, $\Lambda_B$, $\Sigma$, $\epsilon$: error, $M_1,M_2$: Samples}
Initialize: $T_n(d_n)=1,T^2_n(d_n)=1$, $n=1,2$, $d_{prev1}=0,d_1=2\epsilon$, $d_{prev2}=0,d_2=2\epsilon$\\
$P_X=P1,P2$ or $P3$\\
\While{$\vert d_1-{d_{prev1}}\vert>\epsilon$ or $\vert d_1-{d_{prev1}}\vert>\epsilon$} 
{
Calculate $\overline {T_G}$, $\overline {T^2_G}$, $\overline {T_B}$, $\overline {T^2_B}$ as in (\ref{eq:mix_serve})\\
%$\overline {T_G}=T_1(d_1)+T_2(d_2)/(C-1)$\\ 
%$\overline {T^2_G}=T^2_1(d_1)+T^2_2(d_2)/(C-1)$\\
%
%$\overline {T_B}=T_1(d_1)(C-1)+T_2(d_2)$\\ 
%$\overline {T^2_B}=T^2_1(d_1)(C-1)+T^2_2(d_2)$\\
$d_{prev1}\leftarrow d_1$, $d_{prev2}\leftarrow d_2$\\
$d_{old}=0$\\
\While{$\vert d_1-{d_{old}}\vert>\epsilon$}
{
$d_{old}\leftarrow d_1$,\\ $d_1\leftarrow f_G(d_1,S,\overline{T_G}, \overline{T_G^2})$
}
$d_{old}=0$\\
\While{$\vert d_2-{d_{old}}\vert>\epsilon$}
{
$d_{old}\leftarrow d_2$,\\ $d_2\leftarrow f_B(d_2,C,S,\overline{T_B}, \overline{T_B^2})$
}
Use (\ref{eq:ud}), $\Lambda_G,\ \Lambda_B$ and calculate $\pi^{d_n}_S$ for the new $d_n$, $n=1,2$\\ 
\For{$n=1$ \KwTo $2$}
{
\For{$m=1$ \KwTo $M_n$}
{
Sample $U_s,s\in[S]$ using $\pi^{d_n}_S$\\
Sample $\textbf{H}$ (Channel Gains) using $\Sigma$\\
Solve $P_X$ for $U_s,s\in[S]$, $\textbf{H}$ and obtain the service time $T^*$\\
$T_m\leftarrow T^*$
}
${T_n(d_n)}\leftarrow \frac{\sum_{m\in[M_n]}{T_m}}{M_n}$, ${T^2_n(d_n)}\leftarrow \frac{\sum_{m\in[M_n]}{T_m^2}}{M_n}$\\
}
}
$\lambda_G=\sum_{i\in[N],\ j\in{\mathcal{K}_G}}\lambda_{ij}$, 
$\lambda_B=\sum_{i\in[N],\ j\in{\mathcal{K}_B}}\lambda_{ij}$\\
$ $\\s
/*Calculate Mean Sojourn Times $D_1$ for SMQ-G and $D_2$ for SMQ-B*/\\
${D_1}=\frac{\lambda_G'}{\lambda_G} d_1 + \frac{\lambda_G-\lambda_G'}{\lambda_G} \frac{d_1}{2} +{T_1(d_1)} $\\
${D_2}=\frac{\lambda_B'}{\lambda_B} d_2 + \frac{\lambda_B-\lambda_B'}{\lambda_B} \frac{d_2}{2} +{T_2(d_2)} $\\
\KwOut{${D_1}, D_2$: Mean Sojourn Times}
\end{algorithm}

\subsection{Evaluation of Algorithm \ref{algo:mst_miso_dsmq}}
In this section we evaluate the accuracy of the approximation of mean sojourn time provided by Algorithm \ref{algo:mst_miso_dsmq}. Towards this we consider a system with $K=40$ users among which $K_G=20$ are good users and the rest $K_B=20$ are bad users. The mean fading for good users is $g_k=0 dB,\ \forall\ k\in[K_G]$ and for bad users is $g_k=-15 dB,\ \forall\ k\in [K]\setminus[K_G]$. We fix $\lambda=40$. The rest of the parameters are same as in Section \ref{sec:eval_algo1}. We also vary the system parameter $S=1,2$. The algorithm parameters are fixed as follows: $P_X=P1$, $\epsilon=0.1$ and $M_1=500,\ M_2=500$. The comparison here is shown for MMF beamforming (P1). Other beamforming cases (P2 and P3) have similar performance. 

From Figure \ref{fig:Th_Sim_DSMQ} we see that Algorithm \ref{algo:mst_miso_dsmq} (Theory) gives a mean sojourn time with less than $10\%$ error for the case $\lambda=40$. Here we have considered only the high load condition. The Algorithm \ref{algo:mst_miso_dsmq} performs similar to Algorithm \ref{algo:mst_miso_smq} in low load conditions as well, however we do not present here for clarity of presentation.

\begin{figure}[!h]
\centering
\includegraphics[height=6cm,width=8cm,trim={.5cm 7cm 1cm 7cm},clip]{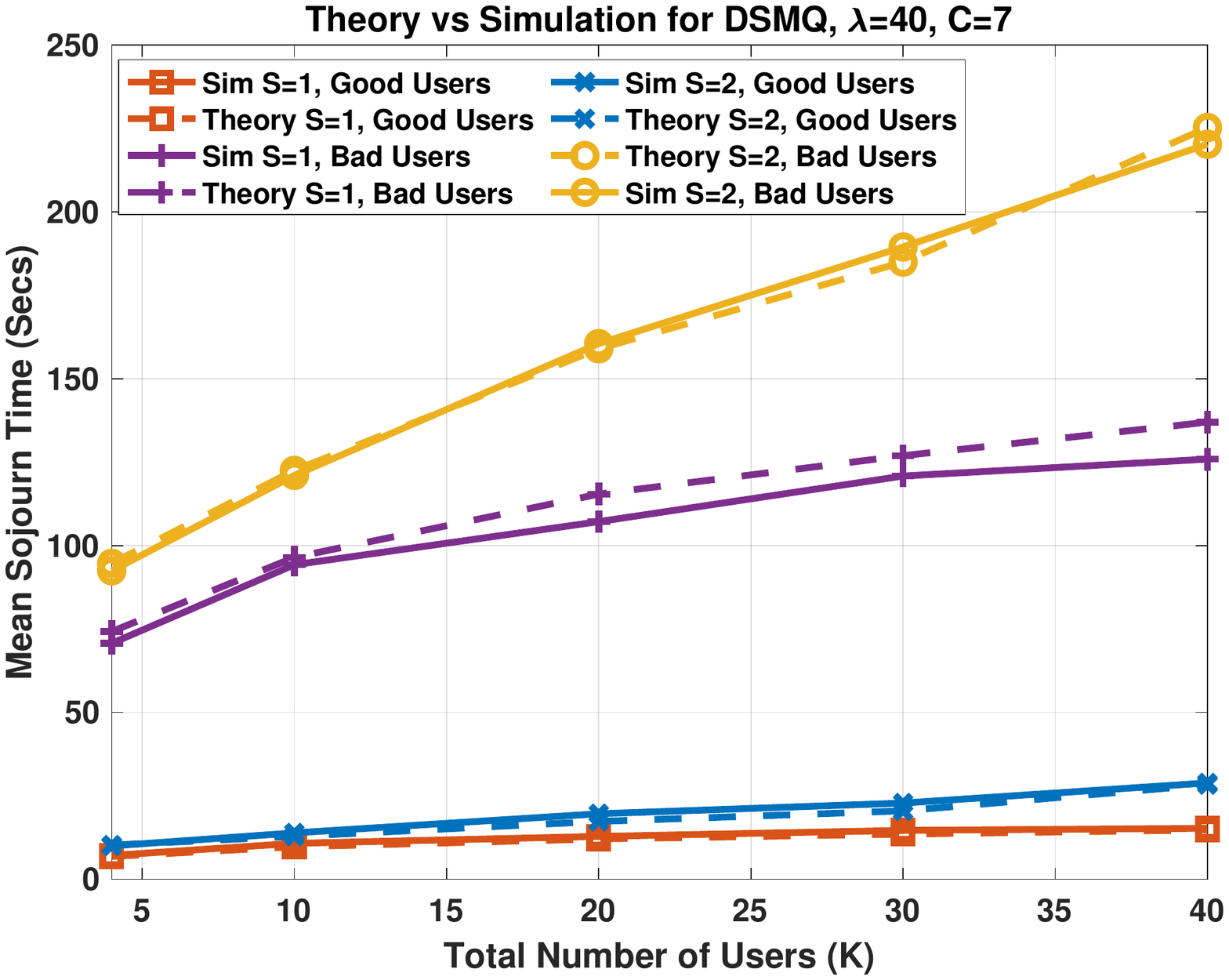}
\caption{Comparison of Theory vs Simulation for DSMQ (Heterogeneous case): $L=16, \lambda=40,\ K=40,\ K_B=K_G=20,\ P=10,\ N=100,\ N_0=1,\ \gamma=1$}
\label{fig:Th_Sim_DSMQ}
\end{figure}

\textit{Note:} \CG We point out that we are not aware of any other study providing theoretical approximations of such complicated queueing systems. Our approximations provide %We point out that the Algorithms \ref{algo:mst_miso_smq} and \ref{algo:mst_miso_dsmq} provides 
a theoretical backing for our queues and provide a means to evaluate the performance of queues for varied system settings without performing full system simulation which can be computationally quite expensive and may take longer to reach stationarity depending on the setting. Our algorithms on the other hand approximate the stationary user distribution in the queue, use closed form expressions for mean delays (Type-1,2 and Mean sojourn time) and typically reach the fixed point within $4$ iterations for all settings. Thus the algorithms are quite useful in system design where engineers may need to do a quick system evaluation with different types of queues. These theoretical approximations also provide insight into how the system behaves and the effect of different parameters on the system performance.

Further our Algorithms are not restricted to beamforming cases presented in this work but can be adapted to any beamforming strategies a designer may wish to apply by appropriately changing the optimization. It would also be much faster if one could get a closed form expressions for the moments of service time for a given user distribution (which typically is not the case in the current multi-user MIMO literature).

We would also like to point out that the Algorithm 2 can be used to analyse any M/G/1 type E-limited Polling system with zero switchover times to get the theoretical waiting times for users, which maybe of independent interest\CE.

%In the following section we present the numerical evaluations of the performance our queues in MISO setup\CE.

\section{Simulation Results and Discussion}
\label{sec:sim}
In this section we present, simulation results and comparison of different beamforming schemes with SMQ and DSMQ. %Among all the physical layer schemes, rate splitting has been shown to provide optimal transmission rate among all schemes in both heavy and light traffic scenarios (\cite{RS2User,RS2017}). We revisit this claim from queueing perspective and compare the mean sojourn time of the proposed schemes including rate splitting. %A key difference is that the earlier studies consider scenarios with a fixed number of active users and request traffic, where as, 
%In contrast to earlier studies we consider the dynamic scenario of varying user groups and traffic as seen in practical networks. Further, the effect of common users across groups and queueing delays is not seen in the above mentioned studies. \CE In the following 
We consider two cases of channel statistics. First with homogenous channel statistics across users and second with heterogenous channel statistics where there are users with good and bad channel statistics\CE. All our simulations use complex Gaussian flat fading channels as explained in Section \ref{sec:system_model}. To avoid arbitrarily large service times, we fix $r_\epsilon=0.01$, (see Section \ref{sec:station}). %, i.e., if the rate given by (\ref{eq:opt}), (\ref{eq:opt_SIC}) or (\ref{eq:opt_RS}) is less than $0.01$, then we assume the rate is $0.01$. However, we have choosen this rate such that the event of (\ref{eq:opt}), (\ref{eq:opt_SIC}) or (\ref{eq:opt_RS}) giving this rate is extremely rare. This is another way of truncating the absolute value of channel fading near zero. 
This ensures that $E[T] <\infty$, needed in Proposition \ref{prop:station}, for all our schemes.  All our simulations are run for $10000$ services of the queue at the BS, to let the queues reach stationarity. \textit{\CK The mean sojourn times are estimated using sample average of sojourn times seen during the simulation}\CE. 
\begin{figure}[h!]
\centering
\includegraphics[height=6cm,width=8cm,trim={1cm 7cm 1cm 7cm},clip]{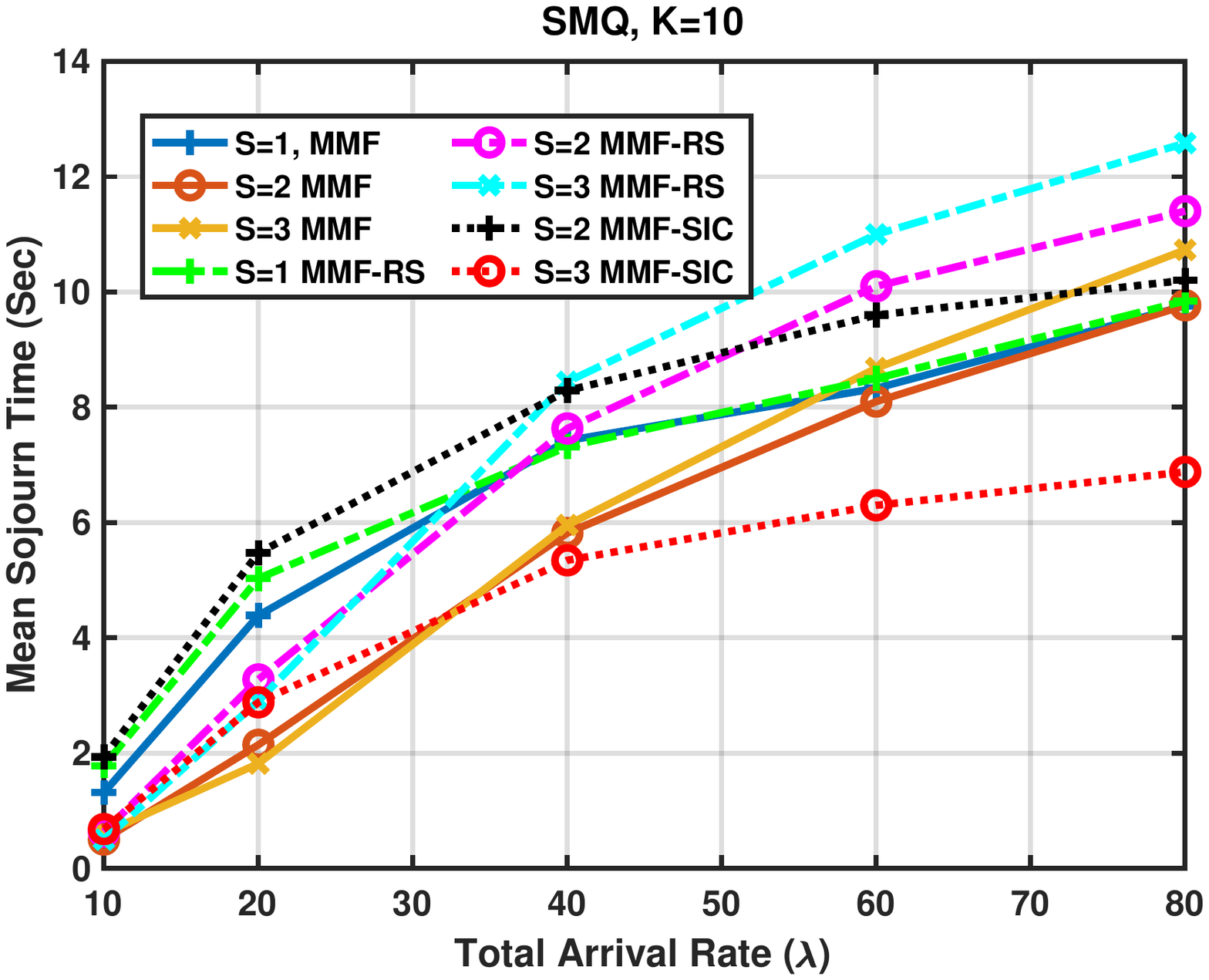}
\caption{Comparison of Beamforming schemes (Homogeneous Case) with SMQ: $K=10,\ P=10,\ N=100,\ N_0=1,\ \gamma=1,\ g=1$}
\label{fig:user_10}
\end{figure}
%%%%%% Commenting 20 User Figure for ICC 
%\begin{figure}[!h]
%\centering
%\includegraphics[height=6cm,width=8cm,trim={1cm 7cm 1cm 7cm},clip]{Figures_PIMRC/20Users_Single_Statistics_RS_Panjus_Results.pdf}
%\caption{Comparison of Beamforming schemes (Homogeneous Case) with SMQ: $K=20,\ P=10,\ N=100,\ N_0=1,\ \gamma=1,\ g=1$}
%\label{fig:user_20}
%\end{figure}
\begin{figure}[h!]
\centering
\includegraphics[height=6cm,width=8cm,trim={1cm 7cm 1cm 7cm},clip]{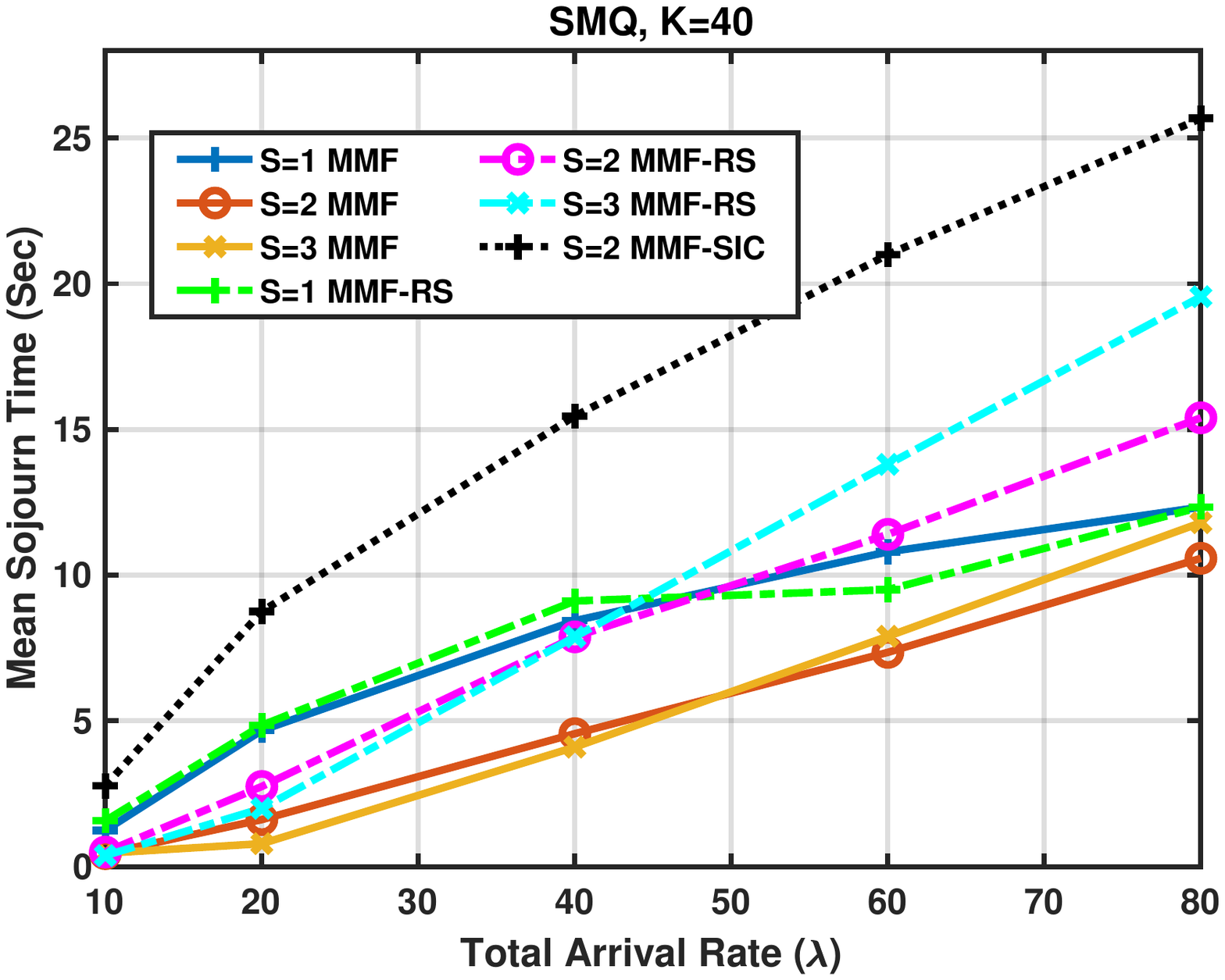}
\caption{Comparison of Beamforming schemes (Homogeneous Case) with SMQ: $K=40,\ P=10,\ N=100,\ N_0=1,\ \gamma=1,\ g=1$}
\label{fig:user_40}
\end{figure}

\textit{Case 1 (Homogeneous Channel Statistics):} \CE We consider a MISO network with $L=16$ antennas, $N=100$ files, each file of size $F=100Mb$ and system bandwidth $B=100MHz$. Channels between each antenna and users are i.i.d complex Gaussian with mean fading $g=1$. The popularity of files follow Zipf distribution with popularity $\gamma=1$. This is a common assumption and is shown to reflect the content request traffic in servers such as youtube \cite{itube}.  We fix $P=10$ and assume that the average noise power $\sigma_k^2=N_0=1,\ \forall k\in[K]$ in all our systems. To cater for all scenarios (low, medium and high traffic) we consider systems with $K=10, 40$ users, and the arrival rates $\lambda=10,20,40,60,80$ files/sec. The first case of $K=10$ represents low load scenario. Since the number of antennas are more than the number of users, the beamformer has higher degrees of freedom to null the inter-stream interference (if any) in all kinds of traffic. However, to bring in the effect of queueing we also look at the different arrival rates $\lambda=10,20,40,60,80$. The second case of $K=40$ represents moderate to high load condition, depending on $\lambda$. Here the total active users (for each file in the queue) may actually be less than the total number of antennas, when the traffic is low (e.g., $\lambda=10,20$) and greater when the traffic is high $\lambda=40,60,80$. %This phenomenon is more pronounced when $K=40$, which represents the heavy load scenario.   

%%%% Modifying for ICC Removing 20 user references .....
Figures \ref{fig:user_10} and \ref{fig:user_40} show the comparison of mean sojourn times for different schemes for different arrival rates for cases $K=10$ and $40$ respectively. \CE %The conclusions for $K=20$ are similar to that of $K=40$ and are not presented here due to space constraints (refer to \cite{report} for $K=20$)\CE. 
The first observation we make is that, increasing the number of streams ($S\geq2$) is beneficial for SMQ MMF and SMQ MMF-RS, only in low and moderate arrival rates. Compared to $S=1$, both $S=2,3$ provide around $50-75\%$ improvement for $\lambda=10,20$ and $30-50\%$ improvement for $\lambda=40$. At very high traffic $\lambda=80$, increasing $S$ provides negligible gain. The reason is that at very high traffic each file entry in the queue has enough requests to provide multicast opportunities and hence adding more streams provides no advantage. However at lower and medium arrival rates the spatial multiplexing gains are provided by $S=2,3$ in addition to the multicast gain provided by the SMQ. Note that the performance reduction of $S=2,3$ streams is also due to the fact that there may exist common users in $S$ groups which may limit the rate, thus reducing the multiplexing gain. 

\CE Our second observation is that the MMF-SIC is beneficial only when the total number of users are less than the total number of antennas, ($K=10,\ L=16$)\CE. 

Further we make another important observation that MMF-RS beamforming in SMQ performs similar to (or) worse than MMF beamforming case for all cases of $K=10,40$ and $S=2,3$. \CE For $S=1$ the performances of MMF and MMF-RS are similar\CE. This is because of the optimization of the max transmit time (\ref{eq:opt_RS}) among the degraded and designated streams, which is inevitable in queued systems such as SMQ\CG. This result is in stark contrast to what is observed in a system without queues \cite{RS2User,RS2017} where, RS based beamforming always does better than any other beamforming scheme. However we will see in the following that MMF-RS provides significant advantage in the heterogenous case, even in a queued system\CE.%, even with the max transmit time optimization between two types of streams\CE.     

\textit{Case 2 (Heterogeneous Channel Statistics):} \CE %Since these MMF and MMF-RS has similar performances at different arrival rate, load regimes and are better than MMF-SIC in homogenous case, we will consider only these two schemes for analyzing the performance in heterogenous statistics case. 
We consider only MMF and MMF-RS in this section (MMF-SIC performs poorly in heterogenous case as well. Hence, we do not present it here for the sake of clarity of presentation). We consider the system with $K=40$ users among which $K_G=20$ are good users and the rest $K_B=20$ are bad users. The mean fading for good users is $g_k=0dB,\ \forall\ k\in[K_G]$ and for bad users is $g_k=-15dB,\ \forall\ k\in [K]\setminus[K_G]$. \CE In other words bad users undergo 15dB deeper fading than the good users. In practical systems this does happen. All the other parameters stay same\CE. We compare the performances of both SMQ and DSMQ in terms of mean sojourn times experienced by good users in Figure \ref{fig:DSMQG} and bad users in Figure \ref{fig:DSMQB}. 
\begin{figure}
\centering
\includegraphics[height=5.5cm,width=8cm,trim={1cm 7cm 1cm 8.7cm},clip]{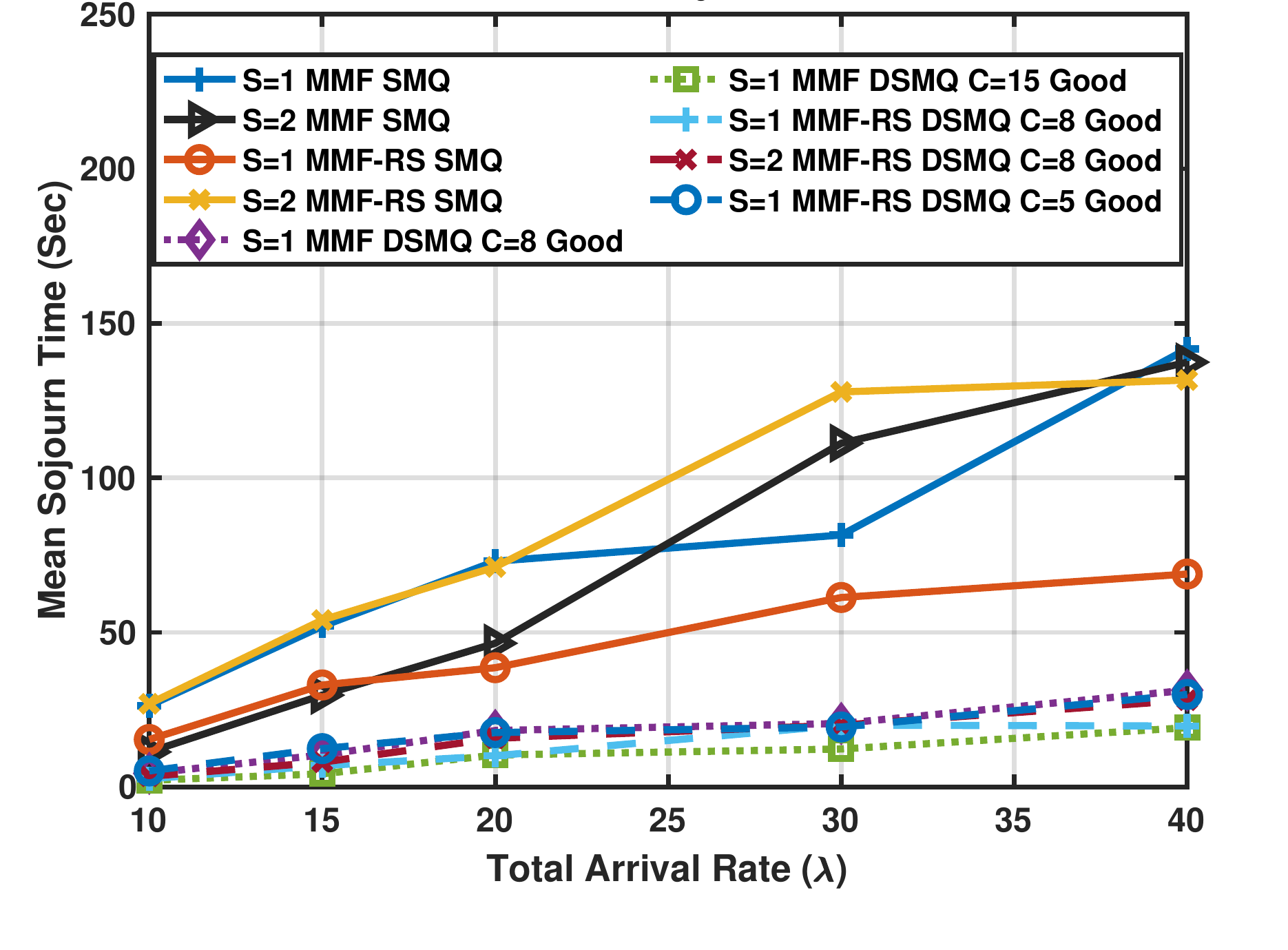}
\caption{Comparison of Beamforming schemes (Heterogeneous Case) with SMQ and DSMQ for Good users: $K=40,\ K_B=K_G=20,\ P=10,\ N=100,\ N_0=1,\ \gamma=1$, }
\label{fig:DSMQG}
\end{figure}
\begin{figure}[!h]
\centering
\includegraphics[height=5.5cm,width=8cm,trim={1cm 7cm 1cm 8.7cm},clip]{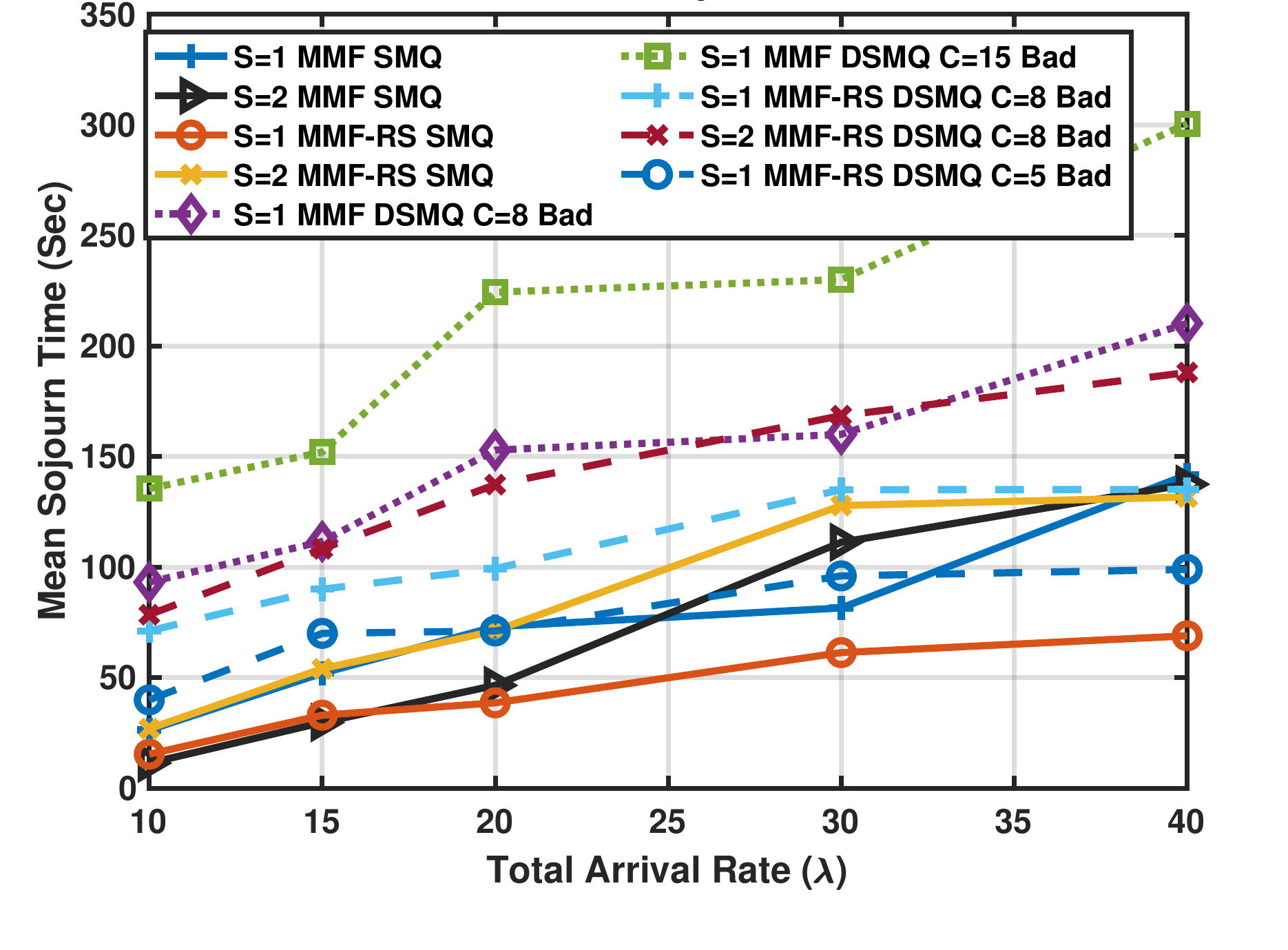}
\caption{Comparison of Beamforming schemes (Heterogeneous Case) with SMQ and DSMQ for Bad users: $K=40,\ K_B=K_G=20,\ P=10,\ N=100,\ N_0=1,\ \gamma=1$}
\label{fig:DSMQB}
\end{figure}

As explained before, in SMQ both good and bad users experience the same mean sojourn time. When we compare MMF SMQ $S=1$ for $\lambda=40$ in Figures \ref{fig:DSMQG}, \ref{fig:DSMQB} with Figure \ref{fig:user_40}, we see that the presence of bad users increases the mean sojourn time of all users (good and bad), from $7$ to $145$ secs. This is a significant degradation\CE. Further, in contrast to the homogeneous case, MMF-RS SMQ $S=1$ improves the mean sojourn time to $70$ secs, in Figure \ref{fig:DSMQG}.  Nevertheless, the delay is still significantly high\CE. Now consider MMF DSMQ with $C=8$, and $C=15$. We see from Figure \ref{fig:DSMQG} that $C=15$ mostly recovers the mean sojourn time for good users to $\sim 20$ secs as compared to $\sim 7$ secs for MMF SMQ in Figure \ref{fig:user_40}. However, the delay of bad users is severely degraded to $300$ secs. This is not desirable. Setting $C=8$ improves this situation by providing mean sojourn time of $30$ secs to good users and $210$ secs to bad users. Thus, $C$ can be tuned to get desirable fairness. 

Next, in Figure \ref{fig:DSMQB} we observe that for $C=8$, MMF-RS DSMQ further reduces the mean sojourn time of bad channel users from $210$ secs for MMF DSMQ to $135$ secs, for $\lambda=40$. This also results in a slight improvement of $5$ secs for the good users, Figure \ref{fig:DSMQG}. %Even though in case of SQM for $\lambda=40$, MMF-RS provides no improvement to MMF for good channel users, see Figures \ref{fig:user_40}, \ref{fig:DSMQG}. The huge gain provided by DSMQ MMF-RS, $C=8$ for bad channel users results in improvement of $5$ secs for good channel users a well. This is due to some residual coupling between SMQ-G and SMQ-B. As seen before, for $C=15$, this coupling reduces with increase in $C$. 
Further, fine tuning of MMF-RS DSMQ with $C=5$ controls the fair allocation of QoS to good and bad channel users, resulting in delays of $30$ and $100$ secs, respectively, for $\lambda=40$. We can see similar trends for $\lambda=10,20$ in Figures \ref{fig:DSMQG} and \ref{fig:DSMQB}. Finally we remark that increasing $S$ ($\geq 2$) is only beneficial in MMF SMQ for $\lambda=10,20$ for reasons similar to the homogeneous case.

\CE Such QoS allocation by DSMQ is quite useful in practical CCN networks to prevent situations where bad users might restrict good users from watching HD content\CG. We now compare DSMQ to other content centric queues like Loopback (\cite{TWC2021}) which was designed for the same purpose (providing optimal fairness).

\textit{DSMQ vs Loopback:} \CG In Figures \ref{fig:DSMQG} we compare the performance of the Loopback scheme proposed in \cite{TWC2021} (Defer scheme \cite{TWC2021} in MISO case also has similar performance. The results are not presented here for the sake of brevity). In the Loopback scheme, we use SMQ and $P1$ MMF beamforming with the following modification. We fix a rate threshold $r_{thresh}$ and transmit at this fixed threshold. We solve $P1$ and obtain rates $R_k$ for all users. Now, the users with rate $R_k\geq r_{thresh}$ only, are successfully served. Rest of the user requests are looped back to the end of the queue. Here, we fix the rate threshold $r_{thresh}=0.5$. This is chosen by trial and error to minimise the overall mean sojourn time. We see in Figure \ref{fig:DSMQG} that Loopback scheme with MMF beamforming, $S=1$, gives a good improvement for the good channel users, for arrival rate $40$, compared to MMF SMQ $S=1$, while slightly degrading the performance for the bad users. However, we note that this improvement is not as good as DSMQ based schemes. Further, we note that Loopback provides no gain for arrival rates $10$ to $30$. This is majorly because the MMF beamforming in a MISO system optimises the transmission rates to users with bad channels. Thus, working with single queues like SMQ and Loopback inherently penalises severely the users with good channels\CG.  

From these simulations it is clear that DSMQ beats SMQ and Loopback's performance in the heterogenous case. Therefore we can conclude that the choice of queueing in a MISO system is an important problem and that the QoS can be significantly improved with careful design.

%Since MMF with $S=1$ performs (OMA) better through out the loaded condition we consider only MMF beamforming scheme with $S=1$, for evaluation of DSMQ. We fix $P=12.5$, $\sigma_k^2=1\forall k$ for all our systems. In this system we have a total of $K=40$ users among which $K_G=20$ are good users and the rest $K_B=20$ are bad users. The mean fading for good users is $g_G=0dB$ and for bad users is $g_G=-15dB$. In other words bad users undergo 15dB deeper fading than the good users. For DSMQ we fix $C=5$. For DSMQ with power control we fix $P_G=3$ and $P_B=40$.

\textit{RL based Power Control in time:} Finally we evaluate the performance of reinforcement learning based power control in time using Adaptive Constrained Deep Q-Network (AC-DQN) algorithm, \cite{Arxiv2019}, in a multi-antenna system with SMQ, for the heterogeneous case for $K=40$ and $\lambda=40$. Towards this, we modify the power constraint in (\ref{eq:opt_re}) as $\sum_{s\in[S]}\norm{\textbf{w}_s}_2^2 \leq P_t$ where $P_t$ is the transmit power of $t^{th}$ transmission. Further we impose an additional long term time average power constraint,  $E[P_t]\leq P$. This constraint is ensured by reinforcement learning based AC-DQN algorithm \cite{Arxiv2019}. %From simulations we see that both AC-DQN and MMF-SMQ without power control achieve similar performance {\CE(see plots in [report reference])}.  
%%%%%%% Commenting figures for ICC%%%%
\begin{figure}[!h]
\centering
\includegraphics[height=5.5cm,width=8cm,trim={1cm 7cm 1cm 7cm},clip]{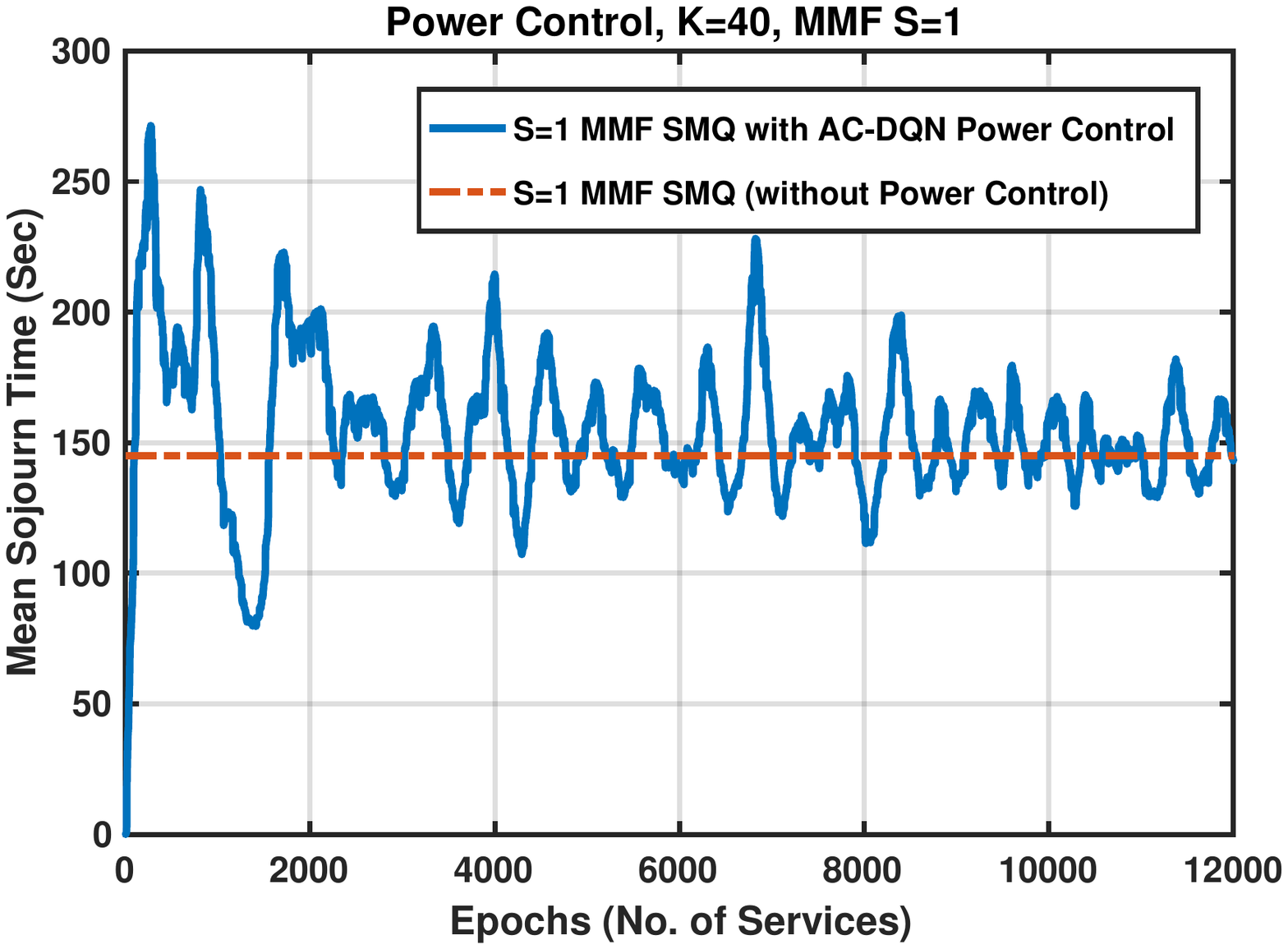}
\caption{ACDQN based power control with SMQ (Heterogeneous Case) MMF S=1: $K=40,\ P=10,\ N=100,\ N_0=1,\ \gamma=1$, $\lambda=40$}
\label{fig:PC}
\end{figure}
%\begin{figure}[!h]
%\centering
%\includegraphics[height=3cm,width=8cm,trim={1cm 11.5cm 1cm 11.5cm},clip]{Figures_NCC/avgpow.pdf}
%\caption{Average Power attained by ACDQN over time (Heterogeneous Case): $K=40,\ P=10,\ N=100,\ N_0=1,\ \gamma=1$, $\lambda=40$}
%\label{fig:avgpow}
%\end{figure}
We see from Figures \ref{fig:PC} that, power control in time for heterogenous user MISO case provides no extra gain. This is in stark contrast to the results obtained in SISO case \cite{Arxiv2019}. This is because spatial diversity in multiple antenna case makes up for most of the time diversity in a SISO systems\CE.

\CR
\textit{Comparison of SMQ with other architectures in Homogeneous Case:}
As mentioned in the introduction, the queueing studies in MIMO literature \cite{huang2012stability,queuestable,Qaware,delayperform,largedevqueue} have invariably considered, different queues for each user. We had argued that this is a suboptimal way of queueing for content-centric networks since such queueing strategies do not utilize the redundancies in the request traffic. In what follows, we show via simulations, that increasing the number of queues to even two in homogeneous case, gives inferior performance, due to reduced multicast opportunities. Therefore, by extension, increasing the number of queues further can only give poorer performance. We show this for a typical case of $K=40$ users and MMF beamforming keeping all other parameters same as in the homogeneous case considered above. The conclusions for other parameters and beamforming cases are similar. 

There are two possibilities, when we consider two queues for the MISO system as mentioned below:
\begin{enumerate}
\item DSMQ MMF (Homogeneous case): Here, similar to the DSMQ queueing scheme, we consider two queues. We call the queues as SMQ1 and SMQ2. We divide the set of users into equal sets and assign SMQ1 for users $\{1,2,\cdots,20\}$ and SMQ2 for users $\{21,22,\cdots,40\}$. Since all users have good channels we set $C=2$, that is equal number of slots are given to both the queues. The beamforming (MMF in this case) and service is provided as explained in Section \ref{sec:DSMQ}, in a slotted manner.
\item 2Q Simultaneous MMF: Another possibility is that, we assign the users as mentioned above and serve the head-of-the-line entries of both the queues simultaneously. That is we consider the set of users in head-of-the-line of SMQ1 and SMQ2, as two groups. We solve optimization problem $P1$ and serve the users in head of the line of both the queues simultaneously as two streams. Here, there are no common users across the two streams. 
\end{enumerate}
Note that to make a fair comparison and give advantage to the two queue systems, we have retained merging in both the queueing systems as described above. Otherwise, as noted in \cite{TWC2021}, without merging, these queues will certainly perform much worse. 

Now, from Figure \ref{fig:opt_user_40}, we see that both the two-queue systems perform poorly compared to SMQ with MMF beamforming. We reiterate that this is because, as we split the queues and the users into groups, the number of multicast opportunities per file go down. As a result the performance is also affected. Hence in a MISO CCN (or any CCN for that matter) it is always a good idea to merge requests for same contents to the maximum extent possible and serve them simultaneously. However, when the channel statistics are heterogenous as seen previously, we also need to consider fairness criterion for good users while performing the merging. Thus the requirement of splitting of queue was valid in the heterogenous case and not so in the homogenous case.  

\begin{figure}[h!]
\centering
\includegraphics[height=5cm,width=8cm,trim={1cm 9cm 1cm 9cm},clip]{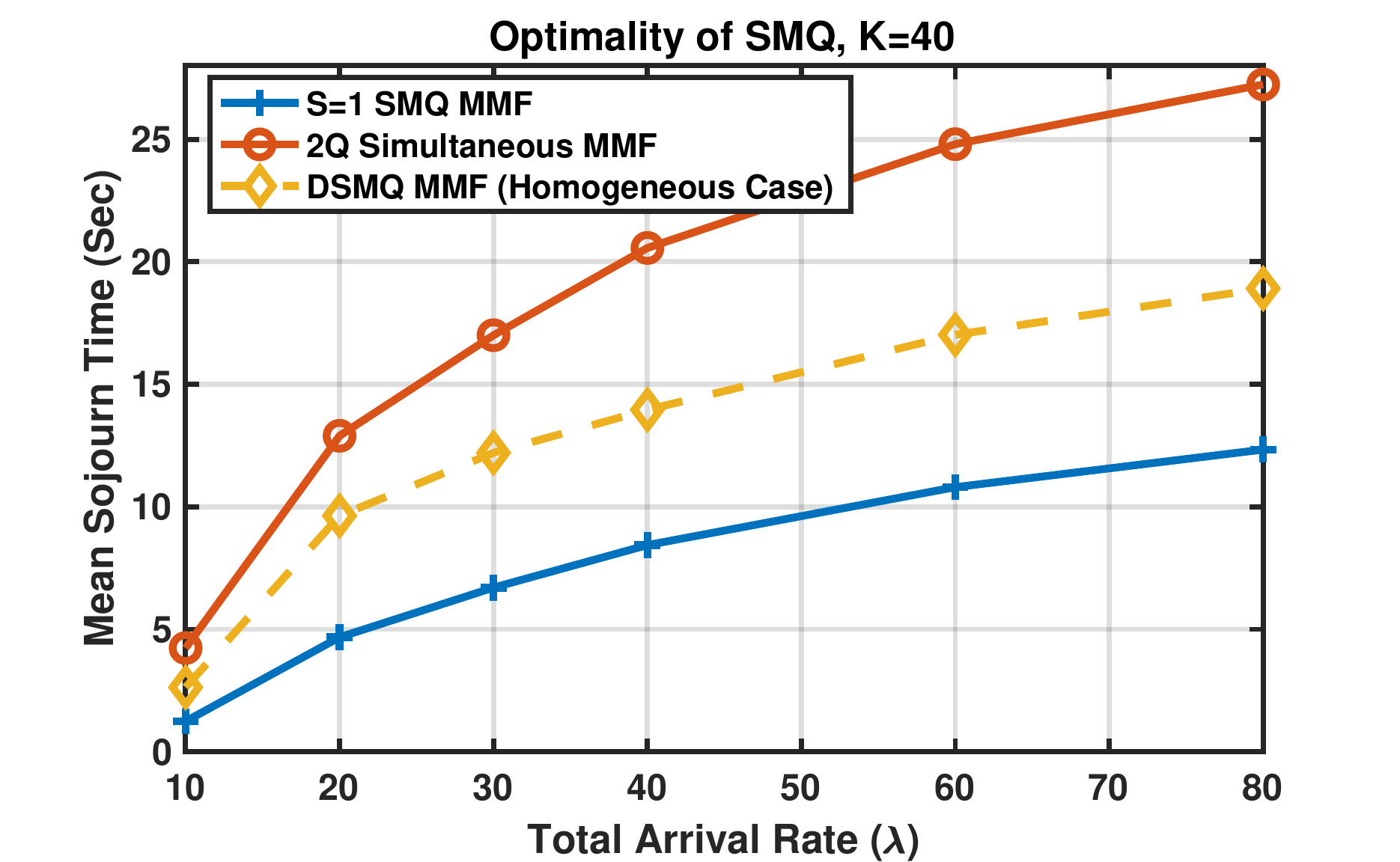}
\caption{Comparison of SMQ with other queues in Homogeneous Case: $K=40,\ P=10,\ N=100,\ N_0=1,\ \gamma=1,\ g=1$}
\label{fig:opt_user_40}
\end{figure}

From these simulations it is clear that MMF-RS DSMQ and SMQ-MMF perform the best in the heterogenous and the homogenous cases, respectively. Therefore we can conclude that the choice of a queueing and a beamforming scheme is a coupled problem and that the QoS and fairness are cross layer objectives which can be significantly improved with careful design.

\section{Conclusion and Future work}
\label{sec:conclusion}
\CE We have considered practical adaptations of beamforming strategies in a MISO CCN with a queue and evaluated their performance.% in a scenario where the BS has a queue. %We have provided an important observation that the performance of the beamforming schemes may drastically vary under such a practical setup, compared to the ones reported in the literature. 
We show that the Simple Multicast Queue (SMQ) can be adapted to such a MISO setup. For homogenous channel case, we show that the SMQ combined with the simplest MMF beamforming scheme performs the best. This is in contrast to the results in \cite{RS2017,RS2User} where Rate Splitting (RS) performs the best. Thus the complexities of RS can be avoided in homogeneous case. 

Further, we have identified SMQ's shortcomings in heterogenous user channel case. We have proposed a new simple queueing scheme called Dual Simple Multicast Queue (DSMQ) which gives flexibility in allocating different QoS for users with good and bad channels. Here, we have shown that MMF-RS DSMQ performs the best, among all schemes. We have also shown that power control and loopback schemes in \cite{TWC2021,Arxiv2019}, are ineffective in MISO setup. %, but, not so when the user channels have similar statistics. Finally we conclude that MMF-RS when combined with DSMQ provides significant advantage of giving good QoS to good channel users without severely penalising users with bad channel. 
\CG We have also proved the stationarity of the queues and have shown that they are always stable across all arrival rates\CE. We have provided queueing theoretic approximations to the mean sojourn time, which could be very useful in cross-layer analysis of Multiuser MISO CCNs with queues. 
Finally we conclude that the selection of the queueing strategy and beamforming is a coupled problem. %and can give drastically different performance in a system with queueing. 
The pairs (SMQ, MMF) and (DSMQ, MMF-RS) are optimal strategies for homogeneous and heterogeneous cases respectively, among the ones considered in this paper\CG. Future work may include analysis of the effect of user movements and imperfect CSIT etc., on queueing delays\CE. 
\bibliographystyle{IEEEtran}
\bibliography{library}
\end{document}